\documentclass[journal, 12pt, onecolumn, draftclsnofoot]{IEEEtran}

\usepackage{epsfig}
\usepackage{graphicx}
\usepackage{mathtools}
\usepackage{amsmath, amsthm, amsfonts}
\usepackage{amssymb, bbm}
\usepackage{booktabs}
\usepackage{siunitx}
\usepackage{cite}
\usepackage{dsfont}
\usepackage{graphics}
\usepackage{xcolor}
\usepackage[inline]{enumitem}
\usepackage{clipboard}
\usepackage{comment}
\newclipboard{myclipboard}

\usepackage{tikz}
\usepackage{tikz-cd}
\usepackage{pgfplots}
\pgfplotsset{compat=1.14}

\newcommand{\Cc}{\mathcal{C}}
\newcommand{\Xc}{\mathcal{X}}
\newcommand{\Yc}{\mathcal{Y}}
\newcommand{\RR}{{\mathbb{R}}}

\newcommand{\Pe}	{P_{\!\rm{e}}}

\newcommand{\Pc}	{P_{\!\rm{c}}}
\newcommand{\Cn} 	{\mathcal{C}_n}

\newcommand{\W}	 	{W}
\newcommand{\Wn} 	{\W^n}
\newcommand{\WW} 	{\boldsymbol{\W}}
\newcommand{\q}	 	{q}
\newcommand{\qn} 	{\q^{n}}
\newcommand{\Zn} 	{Z_n}
\newcommand{\qq} 	{\boldsymbol{\q}}
\newcommand{\PXn}	{P_{X^n}}
\newcommand{\PXnU}	{P_{X^n}^{\Cc_n}}
\newcommand{\PXX}	{\boldsymbol{P}_{\!\XX}}

\newcommand{\PYn}	{P_{Y^n}}

\newcommand{\Mn}	{\mathrm{M}_n}
\newcommand{\Rn}	{\mathrm{R}_n}
\newcommand{\R}		{\mathrm{R}}
\newcommand{\pliminf} {\mathrm{p}\,\text{-}\liminf_{n\to\infty}}

\newcommand{\plimsup} {\mathrm{p}\,\text{-}\limsup_{n\to\infty}}

\newcommand{\Yn}{Y^n}
\newcommand{\Xn}{X^n}
\newcommand{\XX} 	{\boldsymbol{X}}
\newcommand{\YY} 	{\boldsymbol{Y}}
\newcommand{\IIq}	{\underline{I}_{\qq}}
\newcommand{\IIqsup}	{\overline{I}_{\qq}}
\newcommand{\M}		{{\rm K}}
\newcommand{\II}	{\underline{I}}
\newcommand{\IIsup}	{\overline{I}}

\newcommand{\Xbn}{\bar{X}^n}
\newcommand{\yn}{y^n}
\newcommand{\xn}{x^n}
\newcommand{\x} {x}
\newcommand{\CM} 	{C^{\text{sd}}_{\qq}}
\newcommand{\CMrm} 	{C^{\emph{sd}}_{\qq}}
\renewcommand{\Pr}	{\mathbb{P}}

\newtheorem{theorem}{Theorem}
\newtheorem{lemma}{Lemma}

\newtheorem{definition}{Definition}
\newtheorem{proposition}{Proposition}
\newtheorem{corollaryProp}{Corollary}[proposition]

\newcommand{\EX}	{\mathbb{E}}

\DeclareMathOperator*{\argmax}{argmax}

\newcommand{\IND}	  [1]{\mathbbm{1}\{#1\}}

\begin{document}
%
\allowdisplaybreaks

\title{Mismatch Capacity under Stochastic Decoding}

\author{Francesc~Molina and~Albert~Guill\'en i F\`abregas
\thanks{Francesc Molina is with the Department of Signal Theory and Communications, Universitat Polit\`ecnica de Catalunya -- BarcelonaTech, 08034 Barcelona, Spain (e-mail: francesc.molina@upc.edu).
	A. Guill\'en i F\`abregas is with the Department of Engineering, University
of Cambridge, CB2 1PZ Cambridge, U.K., the Department of
Signal Theory and Communications and the Institute of Mathematics (IMTech), Universitat Polit\`ecnica de Catalunya -- BarcelonaTech,
08034 Barcelona, Spain (e-mail: guillen@ieee.org).
}
\thanks{This research was supported by the Spanish Government through MICIU/AEI/10.13039/501100011033 and ERDF/EU under grant PID2025-167766OA-C22, and by the European Research Council under Grant 101142747.
}
\thanks{This research was presented in part at the 2026 IEEE Information Theory Workshop, Arizona (USA), November 2026.}
\thanks{GPT 5.6 Sol (OpenAI) and Claude Opus 5.4 (Anthropic) were used to assist in the proofs of Lemma \ref{lemma:UI}, Proposition \ref{prop:nonzero}, the $\M$-ary symmetric and binary-input binary-output examples, the generation of numerical computations, as well as for overall language checking. All AI-assisted outputs were independently reviewed, verified, improved, explained and written by the authors.}
}

\date{\today}
%

\maketitle

\begin{abstract}
This article investigates channel capacity under mismatched stochastic decoding. Feinstein- and Verdú-Han-type bounds on the error probability of coded communication are derived for sequences of $n$-letter channels and decoding metrics.
These are used to obtain an information-spectrum formula for the corresponding channel capacity: the supremum, over all input distribution sequences, of the limit inferior in probability of the sequence of normalized mismatched information densities.
This formula yields the mismatched analog of the Verdú-Han channel capacity formula.
For discrete-memoryless channels and product metrics, uniform integrability allows to show that the limit product-space improvements of the generalized mutual information and the LM rate give the mismatch capacity under stochastic decoding. This establishes the stochastic decoding counterpart of the Csisz\'ar-Narayan conjecture.
Relationships among the capacity formulas of related decoders are investigated. A number of examples are discussed including the case of binary-input binary-output channels and decoding metrics, where capacity is characterized in single-letter form.
\end{abstract}

\IEEEpeerreviewmaketitle

\newpage
\section{Introduction}
Most coding theorems in information theory are established under the assumption that the channel transition law is known to the decoder. In this setting, the maximum likelihood decoder is the optimal decoder in the sense that it minimizes error probability. 
Other decoders that exploit knowledge of the channel law have also been studied in the literature, such as threshold, typical-set, or stochastic decoders, and have been used to characterize channel capacity.

In situations where the channel law is unknown or costly to estimate, or where maximum likelihood decoding is difficult to implement, decoding is typically based on a more convenient, simpler metric that may be mismatched to the channel statistics \cite{340469,370120,lapidoth1998reliable,Foundations}.
Examples surveyed in \cite{Foundations} include: decoding under channel uncertainty, where only an imperfect channel estimate is available at the decoder; finite-precision  implementations in which the likelihoods are computed using quantized or approximate arithmetic; nearest neighbor decoding in additive non-Gaussian noise channels; 
zero-error and zero-undetected-error communication; and bit-interleaved coded modulation, where independent bit-wise metrics are employed in place of the corresponding symbol metric. In addition, mismatched decoding can be used to derive achievable rates for channels with memory under simplified metrics, as shown for the Gilbert-Elliott channel in \cite{11240395}.

Mismatched decoding defines a variant of the reliable communication problem where the decoder employs a fixed decoding metric \cite{340469,370120,lapidoth1998reliable,Foundations}.
Consider the point-to-point transmission of $\Mn$ messages using the codebook $\Cn = \{\xn(1), \dotsc, \xn(\Mn) \}$ over a channel with transition law $\Wn(\yn|\xn)$ and input-output alphabets $\Xc^n$ and $\Yc^n$, respectively.
In its simplest instance, the channel likelihood $\Wn(\yn|\xn)$ used by the maximum likelihood decoder is replaced by a non-negative metric $\qn(\xn,\yn)$, yielding a decoder that selects the message whose codeword maximizes the decoding metric, i.e., %
\begin{equation}\label{eq:mmdec}
	\hat{m} = \argmax_{1\leq m \leq \Mn} \qn(\xn(m), \yn).%
\end{equation}

The channel capacity under the maximum metric decoder in \eqref{eq:mmdec}, herein referred to as mismatch capacity, is only known as a general information-spectrum formula \cite{HanInfoSpec} as the limit inferior in probability of the sequence of exponents of the pairwise error probability \cite{AneliaGenFormula}. Although very general, the expression has very limited analytical tractability even in the discrete-memoryless case. Indeed, a general single-letter expression for the mismatch capacity remains elusive, apart from a small number of cases, such as for binary-input binary-output channels and metrics \cite[Ex. 6]{370120}.
The literature is rich in lower bounds, among which the generalized mutual information (GMI) \cite{kaplan1993ira} and the LM rate \cite{Hui83,CsiszarKorner81graph} are the most common achievable rates; see \cite{Foundations} and the references therein for a comprehensive survey. In contrast, only few single-letter upper bounds have been proposed \cite{UpperboundEhsan, MulticastAnelia, kangarshahi2022sphere, somekh2024upper}. The GMI and LM rates can be improved by optimizing over $k$-letter input distributions. Csisz\'ar and Narayan conjectured that the mismatched capacity is given by the limit $k$-letter improvement of the LM rate \cite{370120}. The conjecture is only known to be tight for: maximum likelihood decoding or equivalent metrics at $k=1$ \cite{HanInfoSpec}; binary-input binary-output channels and metrics at $k=1$ \cite[Ex.~6]{370120}; erasures-only decoding in the limit $k\to\infty$ \cite{6812221}; and constant-margin decoding as $k\to\infty$ \cite{8365148}.

The difficulty of obtaining a more explicit characterization of the mismatch capacity motivates the study of alternative mismatched decoders based on the same decoding metric, which may admit more tractable expressions.
Specifically, we study channel capacity under a mismatched stochastic decoder which, upon observing $\yn$, randomly decides message $m$ with probability%
\begin{equation}
	P_{\hat{X}^n|\Yn}(\xn(m)| \yn) = \frac{\qn(\xn{(m)}, \yn)}{\sum_{\bar{m}=1}^{\Mn} \qn(\xn{(\bar{m})}, \yn)}.
\end{equation}

Stochastic decoding has a long history in communications theory, with applications spanning both classical and quantum communication. Early applications include its use by Ahlswede and Wolfowitz in the context of arbitrarily varying channels \cite{AHLS_stochastic}.
An analogous stochastic decoder arises naturally in quantum information theory (see, e.g., \cite{holevo2002capacity}) through positive operator-valued measures, which provide the general mathematical framework for quantum measurements. More recently, Yassaee \emph{et al.} employed stochastic decoding to derive general one-shot achievability bounds for single- and multi-user settings \cite{yassaee2013technique}, exploiting the randomized structure to simplify the error analysis.
Scarlett \emph{et al.} derived the ensemble-tight random coding achievable rates, error exponents and elementary properties of the corresponding mismatch capacity \cite{LLDecoder}. In the matched case, stochastic decoding is known to achieve the same capacity and error exponent as maximum likelihood decoding both for i.i.d. and constant-composition ensembles.
An alternative parametrization of the stochastic decoder in \cite{LLDecoder} in terms of joint types was subsequently proposed in \cite{GenStochasticDec}, where random-coding and expurgated error exponents were considered.
Stochastic decoding is also closely related to randomized guessing, in which candidates are drawn from a tilted posterior distribution \cite{MuramatsuMiyakeStochastic,UniversalRandomizedGuessing,MiyamotoNoiseGuessing}.

This work studies the channel capacity adopting a stochastic decoder based on a fixed decoding metric, henceforth referred to as mismatch capacity under stochastic decoding.
We characterize it at three progressively more explicit levels: an information-spectrum formula in terms of the sequence of mismatched information densities; a multi-letter expression obtained by optimizing the sequence of expectations over vector input distributions; and, where the channel-metric structure permits, single-letter characterizations.
More specifically, for $n$-letter channels and decoding metrics, we derive Feinstein- and Verd\'u-Han-type bounds on the error probability and use them to establish an information-spectrum formula, revealing a direct parallel with Verd\'u and Han's general capacity formula \cite{HanInfoSpec,verdu1994general}.
For discrete-memoryless channels and product metrics, we show that the limit product-space improvements of random i.i.d.  and constant-composition codes are capacity achieving.
Our results complement general capacity formulas for related decoding rules \cite{AneliaGenFormula,7464362}, and show that the stochastic version of the Csisz\'ar-Narayan conjecture is tight.
We also establish single-letter expressions for binary alphabets and $\M$-ary symmetric settings.
Indeed, as shown in \cite[Sec. IV]{LLDecoder}, there are situations where the mismatch capacity under stochastic decoding is strictly smaller than that under maximum metric decoding. In some such cases, we are able to characterize the penalty in closed form. For instance, consider the example in \cite[Fig. 1]{LLDecoder}: a binary-symmetric channel (BSC) with crossover probability $p$ with a decoding metric being another BSC with crossover $\delta$. In this case, the mismatch capacity is \cite{370120}
\begin{align}
	C^{\text{mm}}_\q(\W) = \begin{cases}
	1-h(p) &\text{if }\  0<\delta<\frac12\\
	0 & \text{otherwise}
\end{cases}.
\end{align}
We show here that the mismatch capacity under stochastic decoding is given by 
\begin{align}
	C^{\text{sd}}_{\q}(\W) = 
		\begin{cases}
			1-h(p)	&\text{if }\ 0<\delta\leq p\\
			1-h(p)-D(p\Vert\delta)
			&\text{if }\ p<\delta<\frac{1}{2}\\
			0
			&\text{if }\ \frac{1}{2}\leq\delta<1
		\end{cases}.
\end{align}
Observe that in the range $p<\delta<\frac{1}{2}$, there is an explicit penalty $D(p\Vert\delta)$ with respect to maximum metric decoding. An illustration of this is provided in Fig. \ref{fig:BSC}. This expression recovers the example in \cite{LLDecoder} (later used in \cite{MiyamotoNoiseGuessing}) for $p=0.1$ and $\delta=\frac{\sqrt{p}}{\sqrt{p}+\sqrt{1-p}}=0.25$. This example shows that the achievable rate in \cite[Fig. 1]{LLDecoder} is indeed the mismatch capacity under stochastic decoding $C^{\text{sd}}_{\q}(\W)=0.426\ \text{bits/cu}$.

\begin{figure} \centering
\begin{tikzpicture} \small
	\begin{axis}[
		height=7.5cm,
		width=8.5cm,
		xlabel={Metric crossover probability $\delta$},
		ylabel={Capacity [bits/cu]},
		xmin=0, xmax=1, ymin=-0.01, ymax=0.6,
		xtick={0,0.1,0.5,1}, xticklabels={0,0.1,0.5,1},
		ytick={0,0.1,0.2,0.3,0.4,0.5,0.6},
		legend pos=north east, legend cell align=left,
		grid=none, samples=150,
		]
		\addplot[blue,thick,domain=0.003:0.1]
		{(ln(2)+0.1*ln(0.1)+0.9*ln(0.9))/ln(2)};
		\addlegendentry{$C_\q^{\text{sd}}(\W)$}
		
		\addplot[blue,thick,domain=0.1:0.5,forget plot]
		{(ln(2)+0.1*ln(0.1)+0.9*ln(0.9)
			-(0.1*ln(0.1/x)+0.9*ln(0.9/(1-x))))/ln(2)};
		
		\addplot[blue,thick,domain=0.5:0.997,forget plot]{0};
		
		\addplot[red,dashed,thick,domain=0.003:0.5]
		{(ln(2)+0.1*ln(0.1)+0.9*ln(0.9))/ln(2)};
		\addlegendentry{$C^{\text{mm}}_\q(\W)$}
		
		\addplot[red,dashed,thick,domain=0.5:0.997,forget plot]{0};
		
		\draw[dotted] (axis cs:0.1,-0.02) -- (axis cs:0.1,0.5310);
		\draw[dotted] (axis cs:0.5,-0.02) -- (axis cs:0.5,0.5310);
		\draw[dotted] (axis cs:0,0) -- (axis cs:0.5,0);
		
		\node[
		circle,
		draw=red,
		fill=white,
		line width=0.8pt,
		inner sep=1.2pt
		] at (axis cs:0.5,0.5310) {};
		
		\addplot[
		black,
		only marks,
		mark=+,
		mark size=3pt,
		thick
		] coordinates {(0.25,0.426)};
		\node[anchor=west, font=\scriptsize] at (axis cs:0.27,0.426) {(0.25, 0.426)};
		\node[anchor=west, font=\scriptsize] at (axis cs:0.27,0.40) {Ref. \cite{LLDecoder}};
		
		\node[
		circle,
		draw=red,
		fill=red,
		line width=0.8pt,
		inner sep=1.2pt
		] at (axis cs:0.5,0) {};
	\end{axis}
\end{tikzpicture}
\caption{Mismatch capacity under stochastic decoding (solid) and maximum metric decoding (dashed) for $\W=\operatorname{BSC}(0.1)$ and $\q = \operatorname{BSC}(\delta)$. In this case, $C(\W)=0.531$ bits/cu.}
\label{fig:BSC}
\end{figure}

\paragraph*{Notation} Vectors are denoted by lowercase italic letters, with their dimension indicated by a superscript and their components indexed by subscripts: $\x_i$ denotes the $i$-th element of the $n$-dimensional vector $\xn=(\x_1, \dots, \x_n)$. 
Random vectors are denoted analogously by uppercase italic letters: $\Xn=(X_1, \dots, X_n)$ is an $n$-dimensional random vector.
For any distribution $P$ and conditional distribution $V$, we use the equivalent notations $(PV)(y)=\sum_x P(x)V(y|x)=\EX[V(y|X)]$ to denote marginalization.
Boldface notation is used for sequences; see Definition \ref{Def:sequences} for precise definitions.
$|\mathcal{A}|$ denotes the cardinality of set $\mathcal{A}$. The indicator function is denoted by $\mathbbm{1}\{\cdot\}$. Probabilities and expectations are denoted by $\Pr[\cdot]$ and $\EX[\cdot]$, respectively.
Natural logarithms are adopted throughout the paper.
The operators over sequences of random variables in Definition \ref{Def:plim} are used throughout the work.
\begin{definition} \label{Def:sequences}
	We use the following notation:
	\begin{enumerate}[label=\roman*)]
		\item An input process is the sequence of random vectors $\XX \triangleq \{\Xn\}_{n\ge1}$ characterized by the sequence of probability distributions
		\begin{align}
			\PXX \triangleq \{\PXn(\xn)\colon\ \xn\in\Xc^n\}_{n\ge1}.
		\end{align}
		
		\item A channel is defined as the sequence of $n$-letter conditional channel distributions
		\begin{align}
			\WW \triangleq
			\{\Wn(\yn|\xn)\colon\ (\xn,\yn)\in\Xc^n\times\Yc^n\}_{n\ge1}.
		\end{align}
		
		\item A decoding metric is defined as the sequence of non-negative and bounded $n$-letter metrics
		\begin{align}
			\qq \triangleq
			\{\qn(\xn,\yn)\colon\ (\xn,\yn)\in\Xc^n\times\Yc^n\}_{n\ge1}.
		\end{align}		
	\end{enumerate}
\end{definition}
\begin{definition} \label{Def:plim}
	The limit inferior and limit superior in probability of a sequence of random variables $\{Z_n\}_{n\ge1}$ are defined, respectively, by \cite[p. 14]{HanInfoSpec}
	\begin{align} \label{eq:pliminfDef}
		\pliminf Z_n &\triangleq {\sup}\Big\{\alpha\in\RR\colon \limsup_{n\to\infty} \Pr[Z_n\leq\alpha]=0 \Big\}\! 
		\\	\label{eq:plimsupDef}
		\plimsup Z_n &\triangleq {\inf} \Big\{\alpha\in \RR\colon \limsup_{n\to\infty} \Pr[Z_n>\alpha]=0 \Big\}.\!
	\end{align}
\end{definition}

\paragraph*{Paper organization} Section \ref{sec:Prelim} introduces the main concepts and technical background used throughout the manuscript, and provides a technical summary of our contributions.
Section \ref{sec:PeCap} shows upper and lower bounds on the error probability, which are used to obtain a general mismatch capacity formula under stochastic decoding.
Section \ref{sec:Relationships} reveals connections among the mismatch capacities of relevant decoders. The case of discrete-memoryless channels and product metrics is studied in Section \ref{sec:DMCAdd} alongside a number of examples. Proofs of auxiliary lemmas and examples are deferred to the appendices.

\section{Preliminaries and Contributions} \label{sec:Prelim}
We focus on point-to-point reliable communication over a channel $\WW$, where $\Wn\colon\Xc^n{\times}\Yc^n\to[0,1]$ specifies the conditional channel distribution over the input-output alphabets $\Xc^n$ and $\Yc^n$. For each $n$, $\Mn$ equiprobable messages are transmitted using the following $n$-length code.
\begin{definition} \label{Def:Cn}
	An $(n, \Mn, \PXnU)$-code is a set of $\Mn$ messages mapped to distinct codewords $\mathcal{C}_n = \{ \xn{(1)}, \dots, \xn{(\Mn)} \}$, with $\xn{(m)}\in\Xc^n$ for $m=1,\dots,\Mn$. The rate of $\mathcal{C}_n$ is $\Rn = \frac1n \log \Mn$.
	Under equiprobable messages, the induced codebook distribution is $\PXnU(\xn) \triangleq \frac{1}{\Mn}\mathbbm{1}\{\xn \in \mathcal{C}_n\}$.
\end{definition}

Inspired by the work in \cite{LLDecoder}, we adopt a mismatched stochastic decoder driven by a posterior distribution of correct message decoding, as defined below.
\begin{definition} \label{Def:Dec}
	Let $\Cc_n$ be a $(n, \Mn, \PXnU)$-code and let $\qn\colon \Xc^n{\times}\Yc^n \rightarrow [0,\infty)$ be a bounded, non-negative decoding metric.
	Given the channel output sequence $\yn$, the stochastic decoder randomly outputs message $m$ with probability
	\begin{equation} \label{back:PLD}
		P_{\hat{X}^n|\Yn}(\xn(m)| \yn) = \frac{\qn(\xn{(m)}, \yn)}{\sum_{\bar{m}=1}^{\Mn} \qn(\xn{(\bar{m})}, \yn)}.
	\end{equation}
	When $\sum_{\bar{m}=1}^{\Mn} \qn(\xn{(\bar{m})}, \yn)=0$, we adopt the convention that the decoder outputs a message uniformly at random:
	\begin{equation}
		P_{\hat{X}^n|\Yn}(\xn(m)| \yn) = \frac1\Mn
	\end{equation}
	for $m=1, \dotsc, \Mn$.
\end{definition}

The following quantities are central to our analysis.
\begin{definition} \label{Def:C}
    The capacity of a channel $\WW$, denoted by $C(\WW)$, is the supremum of all rates $\R$ for which there exists a sequence of $n$-length codes of size $\Mn$ such that $\liminf_{n\to\infty} \tfrac1n \log \Mn \geq \R$ and arbitrarily small error probability as $n$ tends to infinity.
\end{definition}
\begin{definition} \label{Def:CM}
    The mismatch capacity of a channel $\WW$ under stochastic decoding with metric $\qq$, denoted by $\CMrm(\WW)$, is the supremum of all rates $\R$ for which there exists a sequence of $n$-length codes of size $\Mn$ such that $\liminf_{n\to\infty} \tfrac1n \log \Mn \ge \R$ and arbitrarily small error probability as $n$ tends to infinity.
\end{definition}

Information-spectrum provides the most general characterization of channel capacity in terms of the asymptotic lower-tail behavior of sequences of random variables associated with increasingly long block lengths.
Inspired by Shannon \cite{shannon1957certain}, Verd\'u and Han \cite{verdu1994general} developed this approach for the analysis of the channel capacity $C(\WW)$. For input and output processes $\XX$ and $\YY$, the spectral inf- and sup-information rates are \cite{HanInfoSpec}
\begin{align}
	\II(\XX;\YY)    &\triangleq \pliminf \frac1n \imath(\Xn,\Yn)\\
	\IIsup(\XX;\YY) &\triangleq \plimsup \frac1n \imath(\Xn,\Yn)
\end{align}
where
\begin{align}
	\imath(\xn,\yn) \triangleq \log \frac{\Wn(\yn|\xn)}{(\PXn\Wn)(\yn)} = \log \frac{\Wn(\yn|\xn)}{\PYn(\yn)}
\end{align}
is the information density. Accordingly, $\imath(\Xn,\Yn)$ the $n$-letter information random variable whose expectation is the $n$-letter mutual information.

The general channel capacity formula is given by \cite{verdu1994general}
\begin{align}\label{eq:GenCFormula}
	C(\WW) = \sup_{\PXX} \II(\XX;\YY)
\end{align}
where the supremum is taken over all sequences $\PXX$ according to Definition \ref{Def:sequences}. For discrete-memoryless channels and maximum likelihood decoding, this expression reduces to the known Shannon capacity.
Building on this framework, Somekh-Baruch derived the following general formula for the mismatch capacity \cite{AneliaGenFormula}:
\begin{align} \label{eq:M_Cmm}
	C_{\qq}^{\text{mm}}(\WW) = \sup_{\PXX}\ \pliminf \frac1n \log \frac{1}{\Phi_{\qn}(\Xn,\Yn)}
\end{align}
where $\Phi_{\qn}(\xn,\yn) \triangleq \Pr[\qn(\bar{X}^n,\yn)\geq \qn(\xn,\yn)]$.
Related information-spectrum characterizations were also developed for a broad class of threshold-type decoders in \cite{AneliaGenFormula}, by simply replacing the definition of $\Phi_{\qn}$. For maximum likelihood decoding, the general capacity formula in \eqref{eq:GenCFormula} is recovered.

Our analysis for the stochastic decoder requires introducing the mismatched spectral inf- and sup-information rates
	\begin{align}\label{eq:Def_SpectralInfrate}
		\IIq(\XX; \YY)    &\triangleq \pliminf \frac1n \imath_{\qn} (\Xn, \Yn)\\ \label{eq:Def_SpectralSuprate}
		\IIqsup(\XX; \YY) &\triangleq \plimsup \frac1n \imath_{\qn} (\Xn, \Yn)
	\end{align}	
	defined from the mismatched information density
	\begin{align}\label{eq:mis_inf_dens}
		\imath_{\qn} (\xn,\yn) &\triangleq \log \frac{\qn(\xn,\yn)}{\sum_{\bar\x^n}\PXn(\bar\x^n) \qn(\bar\x^n,\yn)}.
	\end{align}
	We also adopt the convention $\tfrac{0}{0}=1$ in \eqref{eq:mis_inf_dens}.

\begin{lemma} \label{lemma:Iq}
	The following hold for any input process $\XX$ and output process $\YY$ induced by the channel:
	\begin{align} \label{eq:SpectralInfrate}
		\IIq    (\XX;\YY) &\leq \underline{I}(\XX;\YY) \\ \label{eq:SpectralSuprate}
		\IIqsup (\XX;\YY) &\leq \overline{I}(\XX;\YY).
	\end{align}
\end{lemma}
\begin{proof}
	See Appendix \ref{app:Iq}.
\end{proof}

Analogously to the GMI and LM rate \cite{kaplan1993ira, Hui83,CsiszarKorner81graph} for maximum metric decoding, the corresponding random-coding achievable rates under stochastic decoding are obtained by restricting the optimization parameter to $s\in(0,1]$ \cite{LLDecoder}.
These can be improved through $k$-letter extensions obtained by applying the underlying coding construction to blocks of $k$ symbols. This yields the achievable rates%
\begin{align} \label{eq:GMIk_intro}
	I^{\text{sd}, (k)}_{\text{GMI}}(P_{X^k}) &= \sup_{s\in(0, 1]} \EX \biggl[ \log \frac{\q^k(X^k, Y^k)^s}{\EX[\q^k(\bar{X}^k, Y^k)^s|Y^k]}\biggr] \\ \label{eq:LMk_intro}
	I^{\text{sd}, (k)}_{\text{LM}} (P_{X^k}) &= \sup_{\substack{s\in(0, 1]\\a^k(\cdot)}} \EX \Biggl[ \log \frac{\q^k(X^k, Y^k)^s e^{a^k(X^k)}}{\EX[\q^k(\bar{X}^k, Y^k)^s e^{a^k(\bar{X}^k)}|Y^k]}\Biggr].
\end{align}
We further define the corresponding optimizations over input distributions
\begin{align}
C^{\text{sd}, (k)}_{\text{GMI}}&=\sup_{P_{X^k}}I^{\text{sd}, (k)}_{\text{GMI}}(P_{X^k}) \label{eq:cgmik}\\
C^{\text{sd}, (k)}_{\text{LM}} &=\sup_{P_{X^k}}I^{\text{sd}, (k)}_{\text{LM}}(P_{X^k}). \label{eq:clmk}
\end{align}

Motivated by the Csisz\'ar-Narayan conjecture \cite{370120}, a natural question is whether the mismatch capacity under stochastic decoding admits an analogous characterization in terms of the limit product-space improvement of the LM rate:
\begin{equation}
	C_{\q}^{\text{sd}}(\W) = \lim_{k \rightarrow \infty} C^{\text{sd}, (k)}_{\text{LM}}.
\end{equation}
We answer this question in the affirmative for discrete-memoryless channels and bounded, non-negative product metrics that are strictly positive on the channel support.

\subsection{Technical Contributions} \label{sec:Contributions}
We provide several characterizations of the mismatch capacity under stochastic decoding. Our technical contributions are as follows:
\begin{enumerate}
	\item \textit{Probability of error}. We derive Feinstein- and Verdú-Han-type bounds on the probability of error.
	Specifically, for $n\ge1$ and $\gamma>0$, we show the existence of a codebook $\Cc_n$ with $\Mn$ codewords whose error probability is bounded as%
	\begin{equation}
		\Pe(\Cc_n) \leq \Pr\biggl[ \frac1n \imath_{\qn}(\Xn, \Yn) \leq \frac1n\log \Mn +\gamma\biggr] + e^{-n\gamma}
	\end{equation}
	where the probability is taken under the joint law $\PXn \Wn$.
	In addition, for every codebook $\Cc_n$ with $\Mn$ equiprobable codewords the error probability satisfies
	\begin{equation}
		\Pe(\Cc_n) \geq \Pr\biggl[ \frac1n \imath_{\qn}(\Xn, \Yn) \leq \frac1n\log \Mn-\gamma\biggr] - e^{-n\gamma}
	\end{equation}
	 for $\gamma>0$, and where the probability is under $\PXnU\Wn$.
	
	\item \textit{General mismatch capacity formula}. We provide a general formula for the capacity of channel $\WW$ under stochastic decoding with metric $\qq$%
	\begin{equation}
		\CM(\WW) = \sup_{\PXX} \IIq(\XX; \YY).
	\end{equation}
	The result extends the general formula for channel capacity in \cite{verdu1994general}, given in \eqref{eq:GenCFormula}, by replacing the spectral inf-information rate with its mismatched counterpart \eqref{eq:Def_SpectralInfrate}.
	
	\item \textit{Discrete-memoryless channels with product metrics}. We prove the stochastic decoding analog of the Csisz\'ar-Narayan conjecture. Specifically, we show that%
	\begin{align} 
		C^{\text{sd}}_{\q}(\W) &= \lim_{k \rightarrow \infty} C^{\text{sd}, (k)}_{\text{LM}} \\
		&= \lim_{k \rightarrow \infty} C^{\text{sd}, (k)}_{\text{GMI}} \\
		&= \lim_{k \rightarrow \infty} \sup_{P_{\!X^k}} \frac1k \EX \biggl[ \log \frac{\q^k(X^k, Y^k)}{\EX[\q^k(\bar{X}^k, Y^k)|Y^k]} \biggr]%
	\end{align}	
	is achieved via the product-space improvement of either random i.i.d. codes or random constant-composition codes. $C^{\text{sd}}_{\q}(\W)$ equals the multi-letter GMI or LM rate, optimized over increasingly large product spaces.
	We also identify several settings in which stochastic and maximum metric decoding achieve the same capacity, including matched decoding, erasures-only metrics, and channel-metric pairs  where the metric underestimates the channel noise.

	\item \textit{Single-letter expressions}. We identify several examples in which capacity admits a single-letter expression:
	\paragraph{Binary alphabets} For $\Xc=\Yc=\{0,1\}$, we show that random constant-composition codes with $k=1$, that is, without product-space improvement, achieve the mismatch capacity under stochastic decoding:
	\begin{equation}
		C^{\text{sd}}_{\q}(\W) = C_{\text{LM}}^{\text{sd},(1)}.
	\end{equation}
	
	\paragraph{$\M$-ary symmetric channel} We derive an explicit closed-form for $C^{\text{sd}}_{\q}(\W)$: a first regime in which channel capacity is achieved, an intermediate regime in which stochastic decoding incurs a strictly positive capacity loss, and a third regime with zero capacity. The binary case $\M=2$ depicted in Fig.~\ref{fig:BSC} precisely exemplifies this.
\end{enumerate}

\section{Probability of Error and Capacity} \label{sec:PeCap}
In this section, we analyze the stochastic decoder given in Definition \ref{Def:Dec}. The corresponding average probability of correct decoding for an $(n, \Mn, \PXnU)$-code $\Cn$ is given by
\begin{align} \label{PeCap:PcLD1}
	\Pc(\Cn) &= \frac1\Mn \sum_{m=1}^{\Mn} \sum_{\yn}\Wn(\yn|\xn(m))  P_{\hat{X}^n|\Yn}(\xn(m)| \yn)\\
	&= \sum_{\xn} \sum_{\yn} \PXnU(\xn)\Wn(\yn|\xn) \cdot \frac{\qn(\xn, \yn)}{\sum_{\bar{m}=1}^{\Mn} \qn(\xn{(\bar{m})}, \yn)}
	\\ \label{PeCap:PcLD3}
	 &= \EX \Bigg[ \frac{\qn(\Xn, \Yn)}{\sum_{\bar{m}=1}^{\Mn}\qn(\xn(\bar{m}), \Yn)} \Bigg]
\end{align}
where the expectation is taken under the joint law $\PXnU\Wn$. The  average error probability is thus $\Pe(\Cn)=1-\Pc(\Cn)$.

\subsection{Information-Spectrum Formula}

Building on \cite{LLDecoder}, we first establish a Feinstein-type bound for mismatched stochastic decoders in Theorem~\ref{Th:PeAchiev}. 

\begin{theorem} \label{Th:PeAchiev}
	Consider a channel $\WW$ and a decoding metric $\qq$.
	Let $\PXX$ be a sequence of input probability distributions.
	For every $n\ge1$ and $\Mn\ge1$, there exists an $(n, \Mn, \PXnU)$-code $\Cc_n$ whose error probability satisfies
	\begin{equation}
	    \Pe(\Cc_n) \leq \Pr\biggl[ \frac1n\imath_{\qn} (\Xn,\Yn)\leq \frac1n\log\Mn  +\gamma \biggr] + e^{-n\gamma}
	\end{equation}
	where $(\Xn, \Yn) \thicksim \PXn \Wn$, $\gamma>0$ is arbitrary, and $\imath_{\qn}$ is the mismatched information density defined in \eqref{eq:mis_inf_dens}.
\end{theorem}
\begin{proof}
    The proof builds on the random-coding union bound of \cite[Th. 1]{LLDecoder}, which establishes the existence of a code $\Cc_n$ whose error probability is upper-bounded by
    \begin{align}\label{eq:rcu_s}
    	\Pe(\Cc_n) \leq \EX \biggl[ \min \biggl\{ 1, \frac{\Mn{-}1}{s}  \frac{\EX[\qn(\Xbn, \Yn)^s|\Yn]}{\qn(\Xn, \Yn)^s}\biggr\} \biggr]
    \end{align}
	for $s\in(0,1]$. This can be turned into a Feinstein-type bound by following the steps in \cite[Sec. 2.6.4]{Foundations}: first using $\Mn-1\leq \Mn$ and $\min\{1,z\} \leq \IND{z\geq\delta} + \delta$ for any $\delta>0$, and then taking logarithms on both sides. After some algebraic manipulations, we obtain%
    \begin{align}
    	\Pe(\Cc_n) &\leq \EX\biggl[ \mathbbm{1}\biggl\{ \frac{\Mn}{s}  \frac{\EX[\qn(\Xbn, \Yn)^s|\Yn]}{\qn(\Xn, \Yn)^s} \geq \delta \biggr\} \biggr] + \delta\\
    	&= \Pr\biggl[ \frac1n \log\frac{\qn(\Xn, \Yn)^s}{\EX[\qn(\Xbn, \Yn)^s|\Yn]} \leq \frac1n\log \Mn - \frac1n\log s - \frac1n\log \delta \biggr] + \delta.%
		\label{eq:feinstein_s}
    \end{align}
    The desired result follows by setting $\delta = e^{-n\gamma}$, $s=1$, and using the mismatched information density in \eqref{eq:mis_inf_dens} for the term in the left-hand side of the inequality in the probability event.
\end{proof}

We next consider a fixed $(n,\Mn,\PXnU)$-code and derive in Theorem~\ref{Th:PeConverse} a lower bound on the error probability that is reminiscent of the Verd\'u-Han lemma \cite[Th.~4]{verdu1994general}.

\begin{theorem} \label{Th:PeConverse}
	Consider a channel $\WW$ and a decoding metric $\qq$. Let $\{\Cc_n\}_{n\ge1}$ be a sequence of $(n, \Mn, \PXnU)$-codes. For every $n\ge1$, the error probability of $\Cc_n$ satisfies
	\begin{align} \label{Th1:Pe_LD}
		\Pe(\Cn) &\geq \Pr \biggl[\frac1n \imath_{\qn}(\Xn,\Yn) \leq \frac1n\log\Mn - \gamma \biggr] - e^{-n\gamma}
	\end{align}
	where $(\Xn, \Yn) \thicksim \PXnU \Wn$, $\gamma>0$ is arbitrary, and $\imath_{\qn}$ is the mismatched information density defined in \eqref{eq:mis_inf_dens}. 
\end{theorem}
\begin{proof}
	We proceed by rewriting the probability of correct decoding \eqref{PeCap:PcLD3} in terms of the number of codewords $\Mn = e^{n\Rn}$ and the normalized mismatched information density random variable $\Zn\triangleq\tfrac{1}{n}\imath_{\qn}(\Xn, \Yn)$ as%
	\begin{align}\label{T1:1}
		\Pc(\Cc_n) &= \EX \Biggl[ \frac{\qn(\Xn, \Yn)}{\sum_{\bar{m}=1}^{\Mn}\qn(\xn(\bar{m}), \Yn)} \Biggr] \\
		&= \EX \Biggl[ \frac{\qn(\Xn, \Yn)}{\Mn\sum_{\bar\x^n}\PXnU(\bar\x^n) \qn(\bar\x^n,\Yn)} \Biggr] \label{T1:2}\\
		&= \EX \bigl[ e^{n (\Zn-\Rn)} \bigr]\label{T1:3}
	\end{align}
	where $(\Xn, \Yn) \thicksim \PXnU \Wn$. For any $\xi\in\RR$, we first decompose the expectation into two terms according to the events $\{\Zn\leq \xi\}$ and $\{\Zn>\xi\}$ as 
	\begin{align} \label{T1:4}
		\Pc(\Cc_n) &= \EX [ e^{n(\Zn-\Rn)} \IND{\Zn\leq \xi} ] + \EX [ e^{n(\Zn-\Rn)} \IND{\Zn > \xi}],
	\end{align}
	and then upper-bound each term separately. Specifically:%
	\begin{enumerate}[label=\roman*)]
		\item For the first term, since $\Zn\leq \xi$ on the event $\{\Zn\leq \xi\}$, we have%
		\begin{align} \label{T1:5}
			\EX [ e^{n(\Zn-\Rn)} \IND{\Zn\leq \xi} ] 
			&\leq e^{n(\xi-\Rn)} \,\Pr [ \Zn\leq \xi ] \\
			&\leq e^{n(\xi-\Rn)}.
		\end{align}		
		\item For the second term, we use the inequality
		\begin{equation} \label{T1:6}
			e^{n(\Zn-\Rn)} = \frac{\qn(\xn,\yn)}{ \sum_{\bar{m}=1}^{\Mn}\qn(\x^n(\bar{m}), \yn)} \leq 1
		\end{equation}
		for all $\xn\in\Cn$ and  all  $\yn\in\Yc^n$, and consequently
		\begin{equation} \label{T1:7}
			\EX [ e^{n(\Zn-\Rn)} \IND{\Zn> \xi} ]	\leq \Pr [ \Zn> \xi ].
		\end{equation}
	\end{enumerate}	 
	Finally, by choosing $\xi= \Rn - \gamma$ we get%
	\begin{equation} \label{T1:8}
		\Pc(\Cc_n) \leq \Pr\biggl[\frac{1}{n}\imath_{q^n}(\Xn, \Yn) > \Rn - \gamma\biggr] + e^{-n\gamma}.
	\end{equation}
	The claim follows from $\Pe(\Cc_n)=1-\Pc(\Cc_n)$.
\end{proof}

Armed with Theorems \ref{Th:PeAchiev} and \ref{Th:PeConverse}, we derive in the following theorem a general formula for the mismatch capacity under stochastic decoding in Definition \ref{Def:CM}.
\begin{theorem}\label{Th:CM}
	Consider a channel $\WW$ and a decoding metric $\qq$. For each input process $\XX$, let $\YY$ denote the output process induced by transmitting $\XX$ through the channel.
	The mismatch capacity under stochastic decoding is
	\begin{equation}\label{Th:CM_LD}
		\CMrm(\WW) = \sup_{\PXX} \IIq(\XX; \YY)
	\end{equation}
	where the supremum is taken over all input probability distribution sequences $\PXX$ as given in Definition \ref{Def:sequences}. In particular, the following statements hold:
\begin{itemize}
\item(Achievability) For every target rate
\begin{equation}
	\R < \sup_{\PXX} \IIq(\XX; \YY)
\end{equation}
there exists a sequence of codes $\{\Cc_n\}_{n\ge1}$ with $|\Cc_n|=\Mn$ equiprobable codewords,
\begin{align}
\liminf_{n\to\infty} \frac1n\log \Mn \geq \R
\end{align}
and achieving $\limsup_{n\to\infty}\Pe(\Cc_n) =0$.

\item(Converse) Every sequence of codes $\{\Cc_n\}_{n\ge1}$ with $|\Cc_n|=\Mn$ equiprobable codewords and
\begin{equation}
	\liminf_{n\to\infty} \frac1n \log \Mn >\sup_{\PXX} \IIq(\XX; \YY)
\end{equation}
satisfies $\limsup_{n\to\infty} \Pe(\Cc_n)>0$.
\end{itemize}
\end{theorem}

\begin{proof}
	
	\noindent{\em (Achievability)}
	
	Fix a target rate $0\leq \R < \IIq(\XX; \YY)$ and choose $\gamma>0$ such that $\R' \triangleq \IIq(\XX; \YY) - 2\gamma \ge \R$. Define the sequence of codebook sizes $\Mn = \lfloor{e^{n\R'}} \rfloor$  for $n\ge1$ such that $\liminf_{n\to\infty} \tfrac{1}{n}\log \Mn=\R'$. Hence, since $\tfrac1n \log \Mn \leq \R' = \IIq(\XX; \YY) - 2\gamma $, by Theorem \ref{Th:PeAchiev}, there exists an $(n,\Mn,\PXnU)$-code $\Cc_n$ with error probability bounded as
	\begin{align} 
		\Pe(\Cc_n) &\leq \Pr\! \left[ \frac1n \imath_{\qn}(\Xn,\Yn) \leq \frac{1}{n} \log \Mn + \gamma \right ] + e^{-n\gamma}\\
		& \leq \Pr\! \left[ \frac1n \imath_{\qn}(\Xn,\Yn) \leq \IIq(\XX; \YY) - \gamma \right] \! + e^{-n\gamma} \!.
	\end{align}
	Now, taking $\limsup_{n\to\infty}$ on both sides and using the definition of the limit inferior in probability in \eqref{eq:pliminfDef}, we get
	\begin{align}
		\limsup_{n\to\infty} &\ \Pr \!\left[ \frac1n\imath_{\qn}(\Xn,\Yn) \leq\IIq(\XX; \YY) - \gamma \right ] =0
	\end{align}
	and thus $\limsup_{n\to\infty} \Pe(\Cc_n)=0$.
	
	Therefore, every rate $\R<\IIq(\XX;\YY)$ is achievable, and optimizing over $\PXX$ proves the achievability claim.
	
	\vspace{0.5em}
	\noindent{\em (Converse)}
	
	Fix $\gamma>0$ and consider	a sequence $\{\Cc_n\}_{n\ge1}$ such that
	\begin{align}
		\liminf_{n\to\infty} \frac1n\log\Mn > \sup_{\PXX}\IIq(\XX; \YY) + 3\gamma
	\end{align}
	and $\limsup_{n\to\infty} \Pe(\Cc_n)=0$.
	Then, from the definition of limit, there exists an $n_0>0$ such that for $n>n_0$
	\begin{align}
		\frac1n\log\Mn > \sup_{\PXX} \IIq(\XX; \YY) + 2\gamma.
	\end{align}
	Then, from Theorem \ref{Th:PeConverse} and for sufficiently large $n$ we have
	\begin{align}
		\Pe(\Cc_n) \geq&\ \Pr\! \left[ \frac1n \imath_{\qn}(\Xn,\Yn) \leq \sup_{\PXX}\IIq(\XX; \YY) + \gamma \right]\! 
		- e^{-n\gamma}.
	\end{align} 
	Now, by letting $n\to\infty$ and using the definition of the limit inferior in probability in \eqref{eq:pliminfDef}, the right-hand side gives%
	\begin{align}
		\limsup_{n\to\infty} &\ \Pr \!\left[ \frac1n \imath_{\qn}(\Xn,\Yn) \leq \sup_{\PXX}\IIq(\XX; \YY) + \gamma \right] >0
	\end{align}
	and thus $\limsup_{n\to\infty} \Pe(\Cc_n) > 0$. This contradicts the initial assumption $\limsup_{n\to\infty} \Pe(\Cn)=0$.
	The converse follows since $\gamma>0$ was arbitrary, ruling out reliable communication at rates exceeding
	$\sup_{\PXX}\IIq(\XX;\YY)$.
\end{proof}

The following observations related to Theorems \ref{Th:PeAchiev}--\ref{Th:CM} and their proofs are in order.
\begin{enumerate}
\item \textit{Matched stochastic decoding}. When $\qn=\Wn$ for $n\ge1$, Theorem \ref{Th:PeAchiev} provides an analog of Feinstein upper bound while Theorem \ref{Th:PeConverse} readily recovers the Verd\'u-Han lower bound \cite{verdu1994general} for deterministic decoders.
These naturally extend the aforementioned results for deterministic decoders to stochastic decoders. 
In addition, from Lemma \ref{lemma:Iq}, we immediately have the following relationship between general capacity formulas
\begin{equation} \label{eq:CMLD_ML}
	\CM(\WW) \leq C_{\WW}^{\text{sd}}(\WW) = C(\WW).%
\end{equation}
Therefore, stochastic decoding does not pose a limitation to attain the general capacity expression. The right-hand side is also consistent with the fact that, in the discrete-memoryless setting, stochastic and maximum likelihood decoding attain the channel capacity and random-coding error exponent \cite{LLDecoder}.

\item \textit{Strong converse property}. Every sequence of codes with $\liminf_{n\to\infty} \frac1n \log\Mn>\CM(\WW)$ also satisfies $\lim_{n\to\infty} \Pe(\Cc_n) = 1$ iff
\begin{equation}
	\sup_{\PXX} \IIq(\XX; \YY) = \sup_{\PXX} \overline{I}_{\qq} (\XX; \YY).
\end{equation}
This result generalizes the strong converse condition from matched decoding to non-negative bounded metrics under stochastic decoding.
The proof is identical to that in \cite[Th.~3.5.1]{HanInfoSpec} and \cite[Th. 8]{AneliaGenFormula}.

\item \textit{$\varepsilon$-Mismatch capacity}. Theorem \ref{Th:CM} allows the following extension to the $\varepsilon$-mismatch capacity%
\begin{equation}
	C_{\qq}^{\text{sd}, \varepsilon} (\WW) = \sup_{\PXX} \IIq^{\varepsilon} (\XX;\YY)
\end{equation}
where the spectral inf-mismatched information rate is now redefined as
\begin{align}
	\IIq^{\varepsilon} (\XX;\YY) &\triangleq \sup \biggl\{\alpha\in\RR\colon \limsup_{n\to\infty} \Pr\biggl[ \frac1n \imath_{\qn}(\Xn, \Yn)\leq\alpha \biggr] \leq \varepsilon \bigg\}.
\end{align}

\item \textit{Generalized stochastic decoding}. Similarly to \cite{GenStochasticDec}, our results naturally extend to a stochastic decoder operating under a general function applied to metric $\qn$, as $f(\qn)$. The analysis carries over directly upon replacing $\qn(\xn, \yn)$ with  $f(\qn(\xn, \yn))$. This yields the analogs of Theorems \ref{Th:PeAchiev} and \ref{Th:PeConverse} for the error probability by replacing $\qn(\xn, \yn)$ by $f(\qn(\xn, \yn))$.
\end{enumerate}

\subsection{Uniform Integrability and Multi-letter Expressions}

We continue by showing an upper bound to the mismatch capacity formula in Theorem \ref{Th:CM} via the uniform integrability of the sequence of normalized mismatched information densities.
\begin{proposition} \label{Prop:Exp}
	If the sequence of normalized mismatched information densities $\big\{\tfrac1n\imath_{\qn}(\Xn,\Yn)\big\}_{n\ge1}$ is uniformly integrable, then
	\begin{align}
		\pliminf \frac1n \imath_{\qn}(\Xn,\Yn) \leq \liminf_{n\to\infty} \frac1n \EX[\imath_{\qn}(\Xn,\Yn)]
	\end{align}
	where the expectation is taken under the joint law $\PXn\Wn$.
\end{proposition}
\begin{proof}
	We use the shorthand notation: $\Zn = \frac{1}{n}\imath_{q^n}(\Xn,\Yn)$ and $\alpha\triangleq \pliminf \Zn-\delta$ for $\delta>0$. The proof extends the argument used in  \cite[Cor. 2]{AneliaGenFormula} to uniformly integrable $\{\Zn\}_{n\ge1}$.
	We first decompose $\EX[\Zn]$ as 
	\begin{align}
		\EX[\Zn] &= \EX[\Zn \IND{\Zn> \alpha}] + \EX[\Zn \IND{\Zn\le \alpha}]\\
		&\geq \alpha \Pr[\Zn > \alpha] - \EX[|\Zn| \IND{\Zn \leq \alpha}].
	\end{align}
	Taking $\liminf_{n\to\infty}$ on both sides, $\Pr[\Zn> \alpha]\to 1$ and%
	\begin{align} \label{Prop:Exp:P1}
		\liminf_{n\to\infty} \EX[\Zn] &\geq \alpha - \limsup_{n\to\infty}\EX[|\Zn| \IND{\Zn\leq \alpha}].
	\end{align}
	We show that the second term can be made arbitrarily small by uniform integrability of $\{\Zn\}_{n\ge1}$. 
	Specifically, by splitting the term over $\{|\Zn|<K\}$ and  $\{|\Zn|\ge K\}$ \cite[Ch.~13]{williams1991probability}%
	\begin{align}
		\EX[|\Zn| \IND{\Zn\le \alpha}] &= \EX[|\Zn| \IND{\Zn\le\alpha, |\Zn|< K}] + \EX[|\Zn|\IND{\Zn\le \alpha, |\Zn|\ge K}] \\
		&\leq K \Pr[ \Zn\le\alpha] + \EX[|\Zn| \IND{|\Zn|\ge K}]. 
	\end{align}
	Taking $\limsup_{n\to\infty}$ on both sides, using $\Pr[ \Zn\le\alpha]\to 0$ and the uniform integrability of $\{\Zn\}_{n\ge1}$, we have%
	\begin{align}
		\limsup_{n\to\infty}\EX[|\Zn| \IND{\Zn\le \alpha}] &\leq  \epsilon
		\label{eq:uni_int}
	\end{align}
	for any $\epsilon>0$. Hence, substituting it into \eqref{Prop:Exp:P1}, recovering the definition of $\alpha$, and noting that both $\delta>0$ and $\epsilon>0$ are arbitrary, we conclude that%
	\begin{align}
		\liminf_{n\to\infty} \EX[\Zn] &\geq \pliminf \Zn.
	\end{align}
\end{proof}

A direct consequence of Proposition \ref{Prop:Exp} is the result below.
\begin{corollaryProp} \label{Cor:Exp}
	For a uniformly integrable sequence $\{\tfrac1n \imath_{\qn}(\Xn, \Yn)\}_{n\ge1}$, the mismatch capacity under stochastic decoding is upper-bounded as%
	\begin{align}
		\CMrm(\WW) &\leq \lim_{n\to\infty}  \sup_{\PXn} \frac1n \EX[\imath_{\qn}(\Xn, \Yn)] \\
		&= \lim_{n\to\infty} \sup_{\PXn} \frac1n \EX \biggl[ \log \frac{\qn(\Xn, \Yn)}{\EX[\qn(\Xbn, \Yn)|\Yn]} \biggr]
		\label{eq:last_bound}
	\end{align}
	where the inner and outer expectations are taken over $\PXn$ and $\PXn\Wn$, respectively.
\end{corollaryProp}
\begin{proof}
	We first use Proposition \ref{Prop:Exp} to establish
	\begin{equation}
		\CM(\WW) \leq \sup_{\PXX}\ \liminf_{n\to\infty} \frac1n \EX[\imath_{\qn}(\Xn, \Yn)].
	\end{equation}
	Now since $\{\EX[\imath_{\qn}(\Xn, \Yn)]\}_{n\ge1}$ is superadditive and upper-bounded by the matched capacity,  $\{\tfrac1n \EX[\imath_{\qn}(\Xn, \Yn)]\}_{n\ge1}$ has a limit. The claim follows from upper-bounding each term of the sequence indexed by $n$ by taking the supremum over the corresponding set of input distributions.
\end{proof}

The preceding proposition extends the inequality used in the proof of \cite[Cor.~2]{AneliaGenFormula} from sequences of non-negative random variables to uniformly integrable. It also provides a proof of the $\liminf$ counterpart of \cite[Lemma~5.3.2]{HanInfoSpec}. These provide a natural prelude to introduce the notion of mismatched information stability.
\begin{definition} \label{Def:Istability}
	A channel-metric pair $(\WW,\qq)$ is said to be information stable under mismatched stochastic decoding if there exists an input and output process $\XX_{\!\star}$, $\YY_{\!\star}$, such that
	\begin{equation}
		\frac1n \imath_{\qn}(\Xn_{\star}, \Yn_{\star}) \xrightarrow{\emph{prob.}} \lim_{n\to\infty} \frac1n \EX [ \imath_{\qn}(\Xn_{\star},\Yn_{\star}) ]
	\end{equation}
	and
	\begin{align}
		C_{\qq}^{\emph{sd}}(\WW) &= \lim_{n\to\infty} \frac1n \EX [ \imath_{\qn}(\Xn_{\star},\Yn_{\star}) ] .
	\end{align}
\end{definition}
	
\section{Special Cases} \label{sec:Relationships}
\subsection{The Maximum Metric Decoder}
We first consider the connection with the maximum metric decoder. As shown in \cite[Th.~1]{AneliaGenFormula}, the general mismatch capacity formula is%
\begin{equation} \label{Rel:M_Cmm}
    C_{\qq}^{\text{mm}} (\WW) = \sup_{\PXX}\ \pliminf \frac1n \log \frac{1}{\Phi_{\qn}(\Xn, \Yn)}
\end{equation}
where%
\begin{equation} \label{Rel:MM_Phin}
    \Phi_{\qn}(\xn, \yn) \triangleq \Pr[ \qn(\bar{X}^n, \yn) \geq \qn(\xn, \yn)]%
\end{equation}
is the pairwise error probability for input $\xn$ when $\yn$ is received from the channel output.
From Markov's inequality we have
\begin{equation}
    \Phi_{\qn}(\xn, \yn) \leq \frac{\EX[\qn(\Xbn, \yn)]}{\qn(\xn, \yn)} = e^{-\imath_{\qn}(\xn, \yn)}.
\end{equation}
for $n\ge1$. By taking logarithms and $\pliminf$ on both sides, the following holds for every input process:%
\begin{align}
    \pliminf \frac1n \log \frac{1}{\Phi_{\qn}(\Xn, \Yn)} \geq \pliminf \frac1n \imath_{\qn}(\Xn, \Yn).
\end{align}
Therefore, it follows that the capacity under mismatched maximum metric decoding is no less than that under mismatched stochastic decoding:
\begin{align}
    C_{\qq}^{\text{mm}}(\WW) &= \sup_{\PXX}\ \pliminf \frac1n \log \frac{1}{\Phi_{\qn}(\Xn, \Yn)} \\ \label{Rel:MM_Cmm2}
    &\geq \sup_{\PXX}\ \pliminf \frac1n \imath_{\qn}(\Xn, \Yn)\\
    &= \CM(\WW).
\end{align}
This is in line with the results of \cite{LLDecoder} for discrete-memoryless channels and product metrics. The bound in \eqref{Rel:MM_Cmm2} also appears in \cite[Eq. (40)]{AneliaGenFormula} as a lower bound on $C_{\qq}^{\text{mm}}(\WW)$. Our results reveal that, in fact, it corresponds to the mismatch capacity under stochastic decoding.

The previous analysis can be easily extended to the generalized stochastic decoding metric $f(\qn)$, with $f(\cdot)$ a non-decreasing scalar function.
Since Markov's inequality also holds when applying non-decreasing functions of $\qn$ giving
\begin{equation}
   \Phi_{\qn}(\xn, \yn) \leq \frac{\EX[f(\qn(\Xbn, \yn))]}{f(\qn(\xn, \yn))},
\end{equation}
we get%
\begin{equation}
	C^{\text{sd}}_{f(\qq)}(\WW) \leq C_{\qq}^{\text{mm}}(\WW).
\end{equation}
Thus, the resulting stochastic decoder after applying a non-decreasing transformation of the decoding metric cannot outperform  the maximum metric decoder.
 
\subsection{The Stochastic Decoder with Tilting}

A specific case of the above transformation is applying a tilting to the decoding metric; that is, $\qn(\xn,\yn)^{\alpha}$ with $\alpha\geq0$. We consider the following cases:

    \subsubsection{$(n,\alpha)$ with $0<\alpha<\infty$}
  	According to Theorem \ref{Th:CM},%
    \begin{equation} \label{Rel:CM_LDa}
		C_{\qq^\alpha}^{\text{sd}}(\WW) = \sup_{\PXX}\ \pliminf \frac1n \imath_{(\qn)^\alpha}(\Xn, \Yn).
	\end{equation}

    \subsubsection{$(n,\alpha)$ with $1\le\alpha<\infty$}
    We next show that $\alpha\ge1$ lends itself to interesting comparisons. 
   As implied by \eqref{eq:feinstein_s} in the proof of Theorem \ref{Th:PeAchiev}, including the additional tilting $s\in(0,1]$ and optimizing over $s$ does not change the capacity as the supremum is attained for $s=1$. We therefore have
   \begin{align}
		C^{\text{sd}}_{\qq^\alpha}(\WW) &= \sup_{\PXX}\ \pliminf \frac1n \imath_{(\qn)^\alpha}(\Xn, \Yn)\\
        &= \sup_{\PXX} \sup_{s\in(0,1]} \pliminf \frac1n \imath_{(\qn)^{s\alpha}}(\Xn, \Yn)\\
        &\geq  \sup_{\PXX}\ \pliminf \frac1n \imath_{\qn}(\Xn, \Yn)\label{eq:tilted3}\\
        &= \CM(\WW),
	\end{align}
	where \eqref{eq:tilted3} follows from choosing $s=\alpha^{-1}\leq 1$. This shows that the tilting $\alpha\ge1$ can only improve capacity. This effect is observed in \cite[Fig. 1]{MiyamotoNoiseGuessing} with rates and error exponents in the discrete-memoryless case.
	
    \subsubsection{$(n, \infty)$} The case $\alpha\rightarrow\infty$ for every $n\ge1$ turns the stochastic decoder into a maximum metric decoder. Indeed, we have that for $\alpha\to\infty$
    \begin{align}
        P_{\hat{X}^n|\Yn}(\xn(m)| \yn) &= \frac{\qn(\xn(m), \yn)^{\alpha}}{\sum_{\bar{m}=1}^{\Mn} \qn(\xn(\bar{m}), \yn)^\alpha}\\
        &\to \frac{\mathbbm{1}\{\yn \in \mathcal{D}_m\}}{\sum_{\bar{m}=1}^{\Mn} \mathbbm{1}\{\yn \in \mathcal{D}_{\bar{m}}\}}
        \label{eq:likdec_alpha_inf}
    \end{align}
	where 
	\begin{align}
        \mathcal{D}_m \triangleq \big\{\yn \in \Yc^n\colon \qn(\xn(m), \yn)=\max_{1\leq\bar{m}\leq \Mn} \qn(\xn(\bar{m}), \yn)\big\}%
    \end{align}   
    for $m=1,\dots,\Mn$ are the decoding regions of a maximum metric decoder that breaks ties at random. The channel capacity under the decoder in \eqref{eq:likdec_alpha_inf} is the same as that of the maximum metric decoder with disjoint decoding regions that breaks ties as errors \cite[Eqs. (3)-(4)]{AneliaGenFormula}:
    \begin{equation}
    	C^{\text{sd}}_{\qq^\infty}(\WW) = C_{\qq}^{\text{mm}}(\WW).
    \end{equation}

\section{Discrete-Memoryless Channels and Product Metrics} \label{sec:DMCAdd}
This section is devoted to the case where both the $n$-letter channel and decoding metric factorize as products of the single-letter channel $\W$ and metric $\q$; for every $(\xn, \yn) \in \Xc^n {\times} \Yc^n$ we have%
\begin{align}
	\Wn(\yn|\xn)  &= \prod_{i=1}^{n} \W(y_i| x_i) \\
	\qn(\xn, \yn) &= \prod_{i=1}^{n} \q(x_i, y_i).
\end{align}

Our main result is based on the uniform integrability of the sequence of normalized mismatched information densities; a property that holds for a wide class of discrete-memoryless channels and product metrics. 
Alongside Corollary~\ref{Cor:Exp}, this result establishes that both product-space improvements achieve capacity: constant-composition random coding via the multi-letter LM rate, and i.i.d. random coding via the multi-letter GMI.

\newpage
\begin{lemma}\label{lemma:UI}
	For a discrete-memoryless channel $W$ and a bounded single-letter metric $\q$ satisfying%
	\begin{align} \label{eq:lemmaUIcond}
		\min_{(x,y)\colon \W(y|x)>0} \q(x,y) >0,
	\end{align}
	the sequence $\big\{\tfrac1n\imath_{\qn}(\Xn,\Yn)\big\}_{n\ge1}$ is uniformly integrable.
\end{lemma}
\begin{proof}
	Uniform integrability follows from a uniform bound on  $\EX[\imath_{\qn}(\Xn,\Yn)^2]/n^2$; see Appendix \ref{app:ProofUI} for details.
\end{proof}
The condition in Lemma \ref{lemma:UI} is not restrictive as it is satisfied by the class of strictly positive and bounded decoding metrics. 
Zero-undetected error \cite{6812221,AneliaGenFormula,370120} is also covered since the condition allows $\q$ to have zero entries only where $\W$ is zero.

We next show that the upper bound in Corollary \ref{Cor:Exp} is tight when Lemma \ref{lemma:UI} is satisfied. 
\begin{theorem} \label{Th:CM_DMCadd}
	Consider a channel-metric pair satisfying \eqref{eq:lemmaUIcond}. 
	The mismatch capacity under stochastic decoding admits the equivalent expressions:%
	\begin{align} 
		C^{\emph{sd}}_{\q}(\W) &= \lim_{k \rightarrow \infty}  C^{\emph{sd}, (k)}_{\emph{GMI}} \\
		&= \lim_{k \rightarrow \infty}  C^{\emph{sd}, (k)}_{\emph{LM}}
		\\ \label{eq:CM_DMCadd}
		&= \lim_{k \rightarrow \infty} \sup_{P_{\!X^k}} \frac1k \EX \biggl[ \log \frac{\q^k(X^k, Y^k)}{\EX[\q^k(\bar{X}^k, Y^k)|Y^k]} \biggr]
	\end{align}	
	where $C^{\emph{sd}, (k)}_{\emph{GMI}},C^{\emph{sd}, (k)}_{\emph{LM}}$ are defined in \eqref{eq:cgmik} and \eqref{eq:clmk}, respectively.
	The capacity is achieved through product-space improvements of both i.i.d. or constant-composition random coding, yielding the same result as $k\to\infty$ and satisfying mismatched information stability.
\end{theorem}
\begin{proof}
The converse is established in Corollary \ref{Cor:Exp}. The achievability follows from the product-space improvement of random i.i.d. codes:
\begin{align} \label{CMadd:ach1}
	C^{\text{sd}}_{\q}(\W) &\geq \sup_{k\ge1}\ C^{\text{sd}, (k)}_{\text{GMI}} \\ \label{CMadd:ach2}
	&\geq \sup_{k\ge1}\ \sup_{P_{\!X^k}} \frac1k \EX\biggl[ \log \frac{\q^k(X^k, Y^k)}{\EX[\q^k(\bar{X}^k, Y^k)|Y^k]} \biggr]\\ \label{CMadd:ach3}
	&= \lim_{k\to\infty} \sup_{P_{\!X^k}} \frac1k \EX\biggl[ \log \frac{\q^k(X^k, Y^k)}{\EX[\q^k(\bar{X}^k, Y^k)|Y^k]} \biggr].
\end{align}
Equation \eqref{CMadd:ach2} follows by setting $s=1$. 
Equation \eqref{CMadd:ach3} follows from the fact that the supremum over $k\ge1$ and the limit as $k\to\infty$ coincide. This is guaranteed by the superadditivity of  $\{\EX[\imath_{\q^k}(X^k,Y^k)]\}_{k\ge1}$.
For the product-space improvement of the LM rate, we obtain the same result by $s=1$ and $a^k(x^k)=0$ (or any other constant).
\end{proof}

The following observations from Theorem \ref{Th:CM_DMCadd} are in order:%
\begin{enumerate}
	\item \textit{Scaling of $k\equiv k_n$ with block length $n$.} Capacity cannot be achieved for arbitrary scalings of $k_n$ with $n$. The $k_n$-symbol blocks in the product-space construction may grow with $n$, but only slowly enough. For random i.i.d. codes we need $k_n=o(n)$ to use the weak law of large numbers in the achievability \cite{Foundations}. For random constant-composition codes we need $k_n = o(\log n)$ so that the number of type classes remains sub-exponential in $n$.
	
	\item \textit{Optimization over $P_{X^k}$.} The benefit of optimizing over $s$ and $a^k$ vanishes as $k\to\infty$, and optimizing only the multi-letter input distribution suffices to achieve capacity.
	
	\item \textit{Equivalence between the input-optimized GMI and LM rates.} For any $k\ge1$, $C_{\text{GMI}}^{\text{sd},(k)}
	\leq C_{\text{LM}}^{\text{sd},(k)}$ by construction.
	However, the difference vanishes asymptotically,
	as both rates converge to the same limit:
	\begin{equation}
		\lim_{k \rightarrow \infty} C^{\text{sd}, (k)}_{\text{GMI}} = \lim_{k \rightarrow \infty} C^{\text{sd}, (k)}_{\text{LM}}.
	\end{equation}
	This asymptotic equivalence first appeared in a proof in \cite[Sec.~IV-A]{8365148}.
 
	\item \textit{Support of the optimal input distribution}. Both $k$-letter achievable rates are upper-bounded by the cardinality of the input distribution support as
	\begin{equation}
		\frac{I^{\text{sd}, (k)}_{\text{GMI}}(P_{X^k})}{k} \leq \frac{I^{\text{sd}, (k)}_{\text{LM}}(P_{X^k})}{k} \leq \frac{\log |\text{supp}(P_{X^k})|}{k}.
	\end{equation}
	To have positive capacity, the optimal $P_{X^k}$ must have support growing exponentially in $k$ so that the right-hand side does not vanish.
	
	\item \textit{Mismatched information stability}. The condition of mismatched information stability in Definition \ref{Def:Istability} is satisfied for block-i.i.d. input processes for the multi-letter GMI and by input processes uniform over type classes of $k$-length blocks for the multi-letter LM rate.
		
	\item \textit{Positive mismatch capacity}. Paralleling \cite[Th.~2]{370120} and \cite[Th.~4]{LLDecoder}, Proposition~\ref{prop:nonzero} below provides an analog characterization of positive capacity under stochastic decoding, formulated in terms of the input-optimized GMI \eqref{eq:GMIk_intro} for $k=1$. We establish a weaker version of \cite[Th.~4]{LLDecoder} by requiring positivity of the input-optimized GMI rather than the LM rate.
\end{enumerate}

\begin{proposition} \label{prop:nonzero}
	For a channel-metric pair satisfying  \eqref{eq:lemmaUIcond}, $C_\q(\W)>0$ if and only if $C_{\emph{GMI}}^{\emph{sd}}>0$.
\end{proposition}
\begin{proof}
See Appendix \ref{app:nonzero}.
\end{proof}

\subsection{Examples}
We now examine several representative discrete-memoryless channels with product metrics. Some admit single-letter characterizations of mismatch capacity under stochastic decoding, whereas others illustrate the benefit of the multi-letter extension when the channel and metric cannot readily be aligned.

\subsubsection{$\M$-ary symmetric channel and metric}
We consider the $\M$-ary symmetric channel $\W=\operatorname{KSC}(p)$, which generalizes the binary symmetric channel by distributing the total crossover probability $p$ equiprobably among the $\M-1$ output symbols different from the input:
\begin{equation}
	\operatorname{KSC}(p) \triangleq
	\begin{bmatrix}
		1-p
		& \dfrac{p}{\M-1}
		& \cdots
		& \dfrac{p}{\M-1}
		\\[2mm]
		\dfrac{p}{\M-1}
		& 1-p
		& \cdots
		& \dfrac{p}{\M-1}
		\\
		\vdots
		& \vdots
		& \ddots
		& \vdots
		\\[1mm]
		\dfrac{p}{\M-1}
		& \dfrac{p}{\M-1}
		& \cdots
		& 1-p
	\end{bmatrix}.
\end{equation}
The decoding metric is defined analogously as $\q = \operatorname{KSC}(\delta)$. This corresponds to a modulo-$\M$ additive noise channel.

The mismatch capacity under stochastic decoding follows the piecewise structure:%
\begin{align}
	C^{\text{sd}}_{\q}(\W) = 
		\begin{cases}
			C(\W)	& \text{if }\ 0<\delta\leq p\\
			C(\W)-D(p\Vert\delta)
			& \text{if }\ p<\delta<\frac{\M-1}{\M}\\
			0
			& \text{if }\ \frac{\M-1}{\M}\leq\delta<1
		\end{cases}.
\end{align}
In the above equation, $C(\W) = \log \M - (1-p)\log(1-p) - p\log \tfrac{p}{\M-1}$ is the corresponding channel capacity, and $D(p\|\delta)$ is the binary relative entropy. Appendix \ref{app:Mary} contains the proof.

The first regime corresponds to a metric that represents a less noisy version of the true channel. The randomization inherent in stochastic decoding compensates exactly for this underestimation, yielding the matched capacity. The second regime corresponds to an overestimation of the noise level, which cannot be fully corrected because doing so would require sharpening the decoder, or equivalently using a parameter $s>1$. Finally, the third regime corresponds to a complete misalignment between channel and metric. The first and third regimes coincide with the mismatch capacity.

A particular case of this construction is the binary symmetric channel-metric pair obtained for $\M=2$ \cite{MolinaITW}, as described in the introduction. Our results show that the achievable rate computations reported in \cite{LLDecoder} and later in \cite{MiyamotoNoiseGuessing} are indeed the mismatch capacity under stochastic decoding.
Both examples consider $\W_0=\operatorname{BSC}(0.1)$ and $\q_0(x,y) =  \operatorname{BSC}\Big(\frac{\sqrt{0.1}}{\sqrt{0.1} + \sqrt{0.9}}\Big)=\operatorname{BSC}(0.25)$. Fig. \ref{fig:BSC} illustrates precisely this case.
Our result shows that the rate identified by black markers in \cite[Fig.~1]{LLDecoder} and red circles in \cite[Fig.~1]{MiyamotoNoiseGuessing} is indeed the exact mismatch capacity.
\begin{equation}
	C_{\q}^{\text{sd}}(\W) = 1-h(0.1)-D(0.1\|0.25) = 0.426\ \text{bits/cu}.
\end{equation}

\subsubsection{Binary input and output alphabets} \label{sec:BIBO}

Both maximum metric and stochastic decoding capacities admit closely related single-letter LM characterizations and can be achieved by constant-composition random coding without product-space improvments. The single-letter capacity expression is
\begin{align}
	C_\q^{\text{sd}} (\W) = C^{\text{sd}, (1)}_{\text{LM}}.
\end{align} 
The proof is given in Appendix \ref{app:BIBO}.
Under maximum metric decoding, the optimization over $s>0$ yields a sharp transition: the mismatch capacity is either the channel capacity $C(\W)$ or zero \cite{370120}. Our stochastic decoding characterization differs only through $s\in(0,1]$. This limitation introduces an intermediate regime that smooths out the previous sharp transition. 

An important observation from this example is that $C_{\q}^{\text{sd}}(\W)$ admits a single-letter characterization in terms of the LM rate. By contrast, when i.i.d. codes are employed, the GMI requires the multi-letter characterization in Theorem \ref{Th:CM_DMCadd}. This contrast is illustrated in Fig.~\ref{fig:BIBO_lm_gmi}, where a Z-channel is decoded as a BSC. As it can be observed in the figure, even when the metric crossover $\delta\rightarrow0$, the GMI fails to attain capacity. Only as $k$ grows, random i.i.d. codes can attain capacity, as opposed to random constant-composition codes.

\begin{figure}[t!]
	\centering
	\begin{tikzpicture}\small
		\begin{axis}[
			height=7.5cm,
			width=8.5cm,
			xmin=0, xmax=0.60,
			ymin=0, ymax=0.35,
			xlabel={Metric crossover probability $\delta$},
			ylabel={Rate [bits/cu]},
			xtick={0,0.1,...,0.75},
			ytick={0.05,0.10,...,0.35},
			legend style={
				anchor=north east,
				legend columns=1,
			},
			yticklabel style={
				/pgf/number format/fixed,
				/pgf/number format/precision=2,
				/pgf/number format/fixed zerofill
			},
			legend cell align=left
			]
			
			\addplot[
			red, dashed, thick
			] coordinates {
				(0   ,0.3219280944)
				(0.75,0.3219280944)
			};
			\addlegendentry{$C_{\q}^{\text{mm}}(\W)$}
			
			\addplot[black, dashdotdotted, thick, no marks]
			table[x=delta,y expr=\thisrow{LM_sd}/ln(2),col sep=comma]
			{rate_vs_delta_stochastic.csv};
			\addlegendentry{$C_{\q}^{\text{sd}}(\W)$}
			
			\addplot[magenta, thick, no marks]
			table[x=delta,y expr=\thisrow{GMI_k5}/ln(2),col sep=comma]
			{rate_vs_delta_stochastic.csv};
			\addlegendentry{$C_{\text{GMI}}^{\text{sd}, (5)}$}
			
			\addplot[violet, thick, no marks]
			table[x=delta,y expr=\thisrow{GMI_k4}/ln(2),col sep=comma]
			{rate_vs_delta_stochastic.csv};
			\addlegendentry{$C_{\text{GMI}}^{\text{sd}, (4)}$}
			
			\addplot[orange, thick, no marks]
			table[x=delta,y expr=\thisrow{GMI_k3}/ln(2),col sep=comma]
			{rate_vs_delta_stochastic.csv};
			\addlegendentry{$C_{\text{GMI}}^{\text{sd}, (3)}$}
			
			\addplot[teal, thick, no marks]
			table[x=delta,y expr=\thisrow{GMI_k2}/ln(2),col sep=comma]
			{rate_vs_delta_stochastic.csv};
			\addlegendentry{$C_{\text{GMI}}^{\text{sd}, (2)}$}
			
			\addplot[blue, thick, no marks]
			table[x=delta,y expr=\thisrow{GMI_k1}/ln(2),col sep=comma]
			{rate_vs_delta_stochastic.csv};
			\addlegendentry{$C_{\text{GMI}}^{\text{sd}, (1)}$}			
			
			\draw[-stealth, black]
			(axis cs:0.25,0.17) -- (axis cs:0.25,0.27);
			\node[anchor=south, black]
			at (axis cs:0.28,0.27) {$k\!\uparrow\uparrow$};
		\end{axis}
	\end{tikzpicture}	
	\caption{
		Input-optimized $k$-letter GMI versus the BSC metric crossover probability $\delta$ for the Z-channel with crossover probability $0.5$. In this case, $C(\W)=0.322$ bits/cu.
	}	
	\label{fig:BIBO_lm_gmi}
\end{figure}
 
 \subsubsection{Ternary input and output alphabets} \label{sec:TITO}
 We consider the following channel and metric%
\begin{equation} \label{eq:TITO_Wq}
	\W {=}
	\begin{bmatrix}
		0.60 & 0.38 & 0.02 \\
		0.02 & 0.60 & 0.38 \\
		0.38 & 0.02 & 0.60
	\end{bmatrix}
	,\
	\q\ {=}
	\begin{bmatrix}
		0.34 & 0.65 & 0.01 \\
		0.03 & 0.40 & 0.57 \\
		0.64 & 0.02 & 0.34
	\end{bmatrix}
\end{equation}
 aimed at exemplifying the benefit of multi-letter approaches when the channel and metric are not fully aligned. In this specific example, while the channel assigns its largest probability to the diagonal entries, the metric favours a cyclically shifted output with different probabilities.
 
Fig. \ref{fig:TITO_lm_gmi} depicts the optimized multi-letter GMI and LM rates, exhibiting substantial performance gains for $k>1$, although as can be observed from the figure, the sequence is not monotonic in $k$. The sub-sequence of achievable rates for even $k$ in this example is given in Table \ref{tab:lmk}.
The LM rate remains unchanged for even $k\geq 2$, up to numerical precision. This suggests that the two-letter construction captures most of the achievable gain. This construction can be repeated at even block lengths without loss in the achievable rate, whereas odd block lengths are handled by appending a single-letter component, resulting in lower achievable rates.
In contrast, the GMI rate exhibits a small, consistent increase with even $k$ on the order of $10^{-6}$ bits/cu, suggesting that much larger block lengths are required to approach the LM rate.

 \begin{figure}[t]
 	\centering
 	\begin{tikzpicture} \small
 		\begin{axis}[
 			xlabel={Block size $k$},
 			ylabel={Rate [bits/cu]},
 			xmin=1,  xmax=14,
 			ymin=0.35, ymax=0.43,
 			xtick={1,2,4,...,14},
 			ytick={0.35,0.37,...,0.55},
 			legend pos=north east,
 			legend cell align=left,
 			legend style={
 				anchor=north east,
 				legend columns=2,
 			},
 			height=7.5cm,
 			width=8.5cm,
 			yticklabel style={
 				/pgf/number format/fixed,
 				/pgf/number format/precision=2,
 				/pgf/number format/fixed zerofill
 			}
 			]
  			
 			\addplot[
 			blue,
 			thick,
 			mark=*,
 			mark size=2pt
 			]			
 			coordinates {
 				(1,0.39067327)
 				(2,0.41509416)
 				(3,0.40695386)
 				(4,0.41509416)
 				(5,0.41020998)
 				(6,0.41509416)
 				(7,0.41160546)
 				(8,0.41509416)
 				(9,0.41238072)
 				(10,0.41509416)
 				(11,0.41287408)
 				(12,0.41509416)
 				(13,0.41321563)
 				(14,0.41509416)
 			};
 			\addlegendentry{$C_{\text{LM}}^{\text{sd, (k)}}$}
 			
 			\addplot[
 			green!50!black,
 			thick,
 			mark=square*,
			dash dot,
 			mark size=2pt
 			]
 			coordinates {
 				(1,0.35926222)
 				(2,0.41348405)
 				(3,0.38731401)
 				(4,0.41348529)
 				(5,0.39635001)
 				(6,0.41348655)
 				(7,0.40082811)
 				(8,0.41348781)
 				(9,0.40357394)
 				(10,0.41348908)
 				(11,0.40526600)
 				(12,0.41349036)
 				(13,0.40652741)
 				(14,0.41349165)
 			};
 			\addlegendentry{$C_{\text{GMI}}^{\text{sd}, (k)}$}
 		\end{axis}
 	\end{tikzpicture}
 	\caption{Input-optimized $k$-letter GMI and LM rates versus $k$ for the channel and metric in \eqref{eq:TITO_Wq}. In this case, $C(\W)=0.499$ bits/cu.}
 	\label{fig:TITO_lm_gmi}
 \end{figure}
 \begin{table}[t!]	\centering
 	\caption{Computed rates for $k=2,4,6,8,10,12,14$.}
 	\label{tab:lmk}
 	\begin{tabular}{
 			S[table-format=2.0]
 			S[table-format=1.8]
 			S[table-format=1.8]
 		}
 		\toprule
 		{$k$} & {$C_{\text{LM}}^{\text{sd},(k)}$}
 		& {$C_{\text{GMI}}^{\text{sd},(k)}$} \\
 		\midrule
 		2  & 0.41509416 & 0.41348405 \\
 		4  & 0.41509416 & 0.41348529 \\
 		6  & 0.41509416 & 0.41348655 \\
 		8  & 0.41509416 & 0.41348781 \\
 		10 & 0.41509416 & 0.41348908 \\
 		12 & 0.41509416 & 0.41349036 \\
 		14 & 0.41509416 & 0.41349165 \\
 		\bottomrule
 	\end{tabular}
 \end{table}
 
\subsubsection{Erasures-only decoding metric} \label{sec:EO}
Consider the metric $\q_{\text{eo}}(x,y) =	\IND{\W(y|x)>0}$. Let
\begin{equation}
	N(\yn) \triangleq \sum_{m=1}^{\Mn} \IND{\Wn(\yn|\xn(m))>0}
\end{equation}
denote the number of codewords compatible with $\yn$. The maximum metric erasures-only decoder outputs the unique message when $N(\yn)=1$ and declares an erasure otherwise, thereby ensuring zero-undetected-error \cite{370120,6812221,AneliaGenFormula}.
The stochastic decoder instead selects uniformly among the compatible codewords, allowing undetected errors if $N(\yn)>1$.

Particularizing Theorem \ref{Th:CM_DMCadd} yields
\begin{align}
	C^{\text{sd}}_{\text{eo}}(\W) = \lim_{k\to\infty} \sup_{P_{\!X^k}} \frac1k \EX \Biggl[ \log \frac{1}{\sum_{\bar{x}^k: \W^k(Y^k|\bar{x}^k)>0}P_{X^k}(\bar{\x}^k)} \Biggr],
\end{align}
which also satisfies $C^{\text{sd}}_{\text{eo}}(\W) = C_{0\text{-u}}(\W)$, with $C_{0\text{-u}}(\W)$ the zero-undetected-error capacity \cite{6812221}. Thus, replacing erasures by uniform random selection does not affect capacity. In particular, if $\W$ has full support, then every codeword is compatible with every output and $C^{\text{sd}}_{\text{eo}}(\W)=0$.

\newpage
\appendices
\section{Proof of Lemma \ref{lemma:Iq}} \label{app:Iq}
This appendix proves the following relationships:%
\begin{align} \label{appR1:SpectralIM}
	\pliminf \frac1n \imath_{\qn}(\Xn, \Yn) &\leq \pliminf \frac1n \imath(\Xn,\Yn)\\%
	\label{appR1:SpectralIMsup}
	\plimsup \frac1n \imath_{\qn}(\Xn, \Yn) &\leq \plimsup \frac1n \imath(\Xn,\Yn).%
\end{align}

The proof is relatively straightforward from writing
\begin{align} \label{appR1:eq2}
	\frac1n \imath_{\qn} (\Xn, \Yn) &= \frac1n \imath(\Xn, \Yn) + \frac1n \log \frac{\qn(\Xn, \Yn)}{\Wn(\Yn|\Xn)} \frac{(\PXn\Wn)(\Yn)}{(\PXn\qn)(\Yn)}.
\end{align}
After taking $\pliminf$ on both sides of \eqref{appR1:eq2} and using Property 4 in \cite[p. 14]{HanInfoSpec}, we have%
\begin{align}
	\pliminf \frac1n \imath_{\qn}(\Xn, \Yn) &\leq \pliminf \frac1n \imath(\Xn, \Yn) \nonumber\\
	&+ \plimsup \frac1n \log \frac{\qn(\Xn,\Yn)}{\Wn(\Yn|\Xn)} \frac{(\PXn\Wn)(\Yn)}{(\PXn\qn)(\Yn)}.
\end{align}
The final result follows from the fact that the latter term is shown to be non-positive. We have%
\begin{align} \label{appR1:npositive1}
	\Pr\biggl[ \frac1n \log \frac{\qn(\Xn,\Yn)}{\Wn(\Yn|\Xn)}\frac{(\PXn\Wn)(\Yn)}{(\PXn\qn)(\Yn)} > \epsilon \biggr] &= \Pr\biggl[ \frac{\qn(\Xn,\Yn)}{\Wn(\Yn|\Xn)} \frac{(\PXn\Wn)(\Yn)}{(\PXn\qn)(\Yn)}> e^{n\epsilon} \biggr]
	\\ \label{appR1:npositive2}
	&\leq e^{-n\epsilon} \cdot \EX\biggl[\frac{\qn(\Xn,\Yn)}{\Wn(\Yn|\Xn)}\frac{(\PXn\Wn)(\Yn)}{(\PXn\qn)(\Yn)} \biggr] 
	\\ \label{appR1:npositive3}
	&= e^{-n\epsilon} 
\end{align}
where \eqref{appR1:npositive2} is obtained from Markov's inequality. Hence, the above probability vanishes as $n\to\infty$ for every $\epsilon>0$, and, according to  $\plimsup$ defined in \eqref{eq:plimsupDef}, the infimum is upper-bounded by zero.
\eqref{appR1:SpectralIMsup} is proved analogously.

\section{Proof of Lemma \ref{lemma:UI}} \label{app:ProofUI}
This appendix proves the uniform integrability of $\{\Zn\}_{n\ge1}$ with $\Zn \triangleq \tfrac1n\imath_{\qn}(\Xn,\Yn)$ for a product metric with a bounded single-letter decoding metric $\q(x,y)$ that satisfies
\begin{align} \label{eq:CondUI}
	\q_{\star} = \min_{(x,y)\colon \W(y|x)>0} \q(x,y) >0.
\end{align}

All we need to prove is that for every $\epsilon>0$, there exists $K>0$ such that
\begin{align}
	\sup_{n\ge1} \EX[|\Zn| \IND{|\Zn|\ge K}] \leq \epsilon.
\end{align}
Uniform integrability is immediate under bounded and strictly positive decoding metrics. Our proof is more involved since we encompass the case where $\q(x,y)=0$ when $\W(y|x)=0$.
Following \cite[Sec. 13.3]{williams1991probability}, we show the above from a uniform bound on the second moments of $\Zn$. Specifically,
\begin{align}
	\sup_{n\ge1} \EX [ |\Zn| \IND{|\Zn| \ge K} ] \leq \sup_{n\ge1} \frac{\EX[\Zn^2]}{K} 
\end{align}
for every $K>0$. Hence, if $\sup_{n\ge1} \EX[\Zn^2]<\infty$, we can always choose $K$ sufficiently large so that the right-hand side is smaller than $\epsilon$, therefore proving that $\{\Zn\}_{n\ge1}$ is uniformly integrable.

We now verify the required second moment bound:%
\begin{align}
	\EX[\Zn^2] &= \frac{1}{n^2} \EX[  \imath_{\qn}(\Xn, \Yn)^2 ] \\
	&= \frac{\EX [ ( \log \qn(\Xn, \Yn) - \log (\PXn\qn)(\Yn))^2 ]}{n^2}\\
	&\leq \frac{2 \EX[ \log^2 \qn(\Xn, \Yn)]}{n^2} + \frac{2\EX[ \log^2 (\PXn\qn)(\Yn) ]}{n^2}
\end{align}
where we have used $(a-b)^2\leq 2a^2+2b^2$.

We bound each term separately, assuming that $\q$ satisfies%
\begin{align}
	0 < \q_{\star} \leq \q(x, y) \leq 1
\end{align}
for every $(x,y)\in \Xc\times\Yc$ with $\W(y|x)>0$.
If $\q(x,y)>1$, we scale $\q$ by a positive constant, so that $\q(x,y)\leq1$; this does not affect the decoder. The bounds on each term rely on the fact that $\log^2 (x)$ is decreasing in $(0,1]$ and thus $\log^2 f(x) \leq \log^2 \min_{x\in(0,1]} f(x)$.
Specifically:%
\begin{enumerate}[label=\roman*)]
	\item For the first term, $\qn(\xn, \yn)\geq \qn_{\star} \cdot \mathbbm{1}\{\Wn(\yn|\xn)>0\}$ and thus%
	\begin{align}
		\frac{2 \EX[ \log^2 \qn(\Xn, \Yn)]}{n^2} &\leq 2\log^2 \q_{\star}. 
	\end{align}
	
	\item As for the second term, $(\PXn\qn)(\yn) \geq \qn_{\star} \cdot\Pr[\Wn(\yn|\Xbn)>0] \geq \qn_{\star} / |\Xc|^n$ and thus
	\begin{align}
		\frac{2\EX[ \log^2 (\PXn\qn)(\Yn) ] }{n^2} & \leq 2 \log^2 \frac{\q_{\star}}{|\Xc|}.
	\end{align}
\end{enumerate}
By combining both results, $\EX[\Zn^2] \leq 2 \log^2 \q_{\star} + 2\log^2 (\q_{\star}/|\Xc|)$. Hence, the statement follows since $\q_{\star}>0$ and the input alphabet is discrete.

\section{Proof of Proposition \ref{prop:nonzero}} \label{app:nonzero}
This appendix proves that $C_\q(\W)>0$ iff $C_{\text{GMI}}^{\text{sd},(1)}(\W)>0$.

Define
\begin{align}\begin{aligned}
	\Delta_{k}(P_{X^k}) &\triangleq {\sum_{x^k, y^k}} P_{X^k}(x^k)\W^k(y^k|x^k) \log \q^k(x^k,y^k) \\
	&- {\sum_{x^k,y^k}} P_{X^k}(x^k) P_{Y^k}(y^k) \log \q^k(x^k,y^k).
\end{aligned}\end{align}
where $P_{Y^k}$ is the output distribution induced by $P_{X^k}$ and channel $\W^k$.
For a discrete-memoryless $\W^k$ and a product metric $\q^k$, the above quantity decomposes as
\begin{equation}
	\Delta_k(P_{X^k}) =	\sum_{i=1}^k \Delta_1(P_{X_i}),
\end{equation}
where $P_{X_i}$ is the $i$-th marginal of $P_{X^k}$.

The direct proof is immediate. Since $C_{\text{GMI}}^{\text{sd},(1)}>0$ is an achievable rate, we immediately have $C_\q^{\text{sd}}(\W)>0$. By construction of the achievable rate, there exists an input distribution $P_X$ such that $I_{\text{GMI}}^{\text{sd},(1)} (P_{X})>0$ and thus $\Delta_1(P_X)>0$.

Conversely, assume that $C_{\q}^{\text{sd}}(\W)>0$. Invoking Theorem \ref{Th:CM_DMCadd}, there exist $k\geq 1$ and an input distribution $P_{X^k}$ for which $I_{\text{GMI}}^{\text{sd}, (k)} (P_{X^k})>0$, and thus $\Delta_k(P_{X^k})>0$. For a memoryless channel and product metric, the latter implies $\Delta_1(P_{X_i})>0$ for some $i$, and consequently $I_{\text{GMI}}^{\text{sd},(1)} (P_{X_i})>0$.

\section{Proof for the KSC Channel and Metric} \label{app:Mary}

We analyze the following three regimes separately, proceeding in increasing order of $\delta$.
\subsection{$0<\delta\leq p$:}
This case corresponds to a decoding metric that underestimates the channel crossover probability. In this case, the GMI \eqref{eq:GMIk_intro} with $k=1$ is optimized for
\begin{equation}
	s_\star = 
		\bigg({\log}\frac{(M-1)(1-p)}{p}\bigg)\bigg({\log}\frac{(M-1)(1-\delta)}{\delta}\bigg)^{-1}
\end{equation}
which is such that $0<s_\star\leq1$.
 Specifically, with $s=s_\star$ we can write $\q(x, y)^{s_\star} = c \cdot \W(y|x)$ with $c>0$, and thus $I_{\text{GMI}}^{\text{sd},(1)}(P_X) =I(P_X;\W)$.
Choosing the equiprobable input distribution , we recover%
\begin{equation}
	C_{\q}^{\text{sd}}(\W) = C(\W) = \log \M - h_\M(p)
\end{equation}
with $h_\M(p)$ the $\M$-ary entropy function. 

\subsection{$p<\delta<\tfrac{\M-1}{\M}$:}
In this case, the decoding metric overestimates the channel crossover probability. Under maximum metric decoding, this mismatch can be compensated for by a tilt $s>1$ which aligns the tilted metric with the channel.
Stochastic decoding instead restricts $s\in(0,1]$; hence, the optimum is attained at  $s_\star=1$. The input-optimized GMI for $k=1$ is%
\begin{equation}
	C_{\text{GMI}}^{\text{sd},(1)} = C(\W)-D(p\|\delta).%
\end{equation}
We next show that $C_{\text{GMI}}^{\text{sd},(1)}$ is the mismatch capacity under stochastic decoding, by deriving the matching upper bound%
\begin{equation}
	\IIq(\XX;\YY) \leq C(\W)-D(p\|\delta).
\end{equation}

The argument begins with the inequality%
\begin{align}
	\IIq(\XX;\YY) &= \pliminf \biggl[ \frac1n \log \qn(\Xn,\Yn) + \frac1n \log\frac{1}{P_{\widetilde Y^n}(\Yn)} \biggr] \\
	&\leq \pliminf \frac1n \log \qn(\Xn,\Yn) + \plimsup \frac1n \log\frac{1}{P_{\widetilde Y^n}(\Yn)}
\end{align}
where $P_{\widetilde Y^n}(\yn)=(\PXn\qn)(\yn)$.

We first establish that%
\begin{equation}\label{eq:app-P1}
	\pliminf \frac1n \log \qn(\Xn,\Yn) = -h_\M(p)- D(p\|\delta)
\end{equation}
by writing the normalized log-metric as%
\begin{align}
	\frac1n \log \qn(\Xn,\Yn) &= \frac1n \sum_{i=1}^{n} \IND{X_i=Y_i}\log(1-\delta) +
	\frac1n \sum_{i=1}^{n} \IND{X_i\neq Y_i} \log\frac{\delta}{\M-1}.
	\label{eq:app-logmetric}
\end{align}
Under the stationary memoryless $\M$-ary symmetric channel, $\IND{X_i= Y_i}$ and $\IND{X_i\neq Y_i}$ for $i=1,\ldots,n,$ are i.i.d. Bernoulli random variables with probabilities $1-p$ and $p$, respectively.
Hence, the law of large numbers gives%
\begin{align}
	\frac1n \log \qn (\Xn, \Yn) \xrightarrow{\text {prob.}} (1-p)\log(1-\delta) + p\log\frac{\delta}{\M-1}.
	\label{eq:app-logmetric-limit}
\end{align}
The convergence in probability to a constant in \eqref{eq:app-logmetric-limit} implies \eqref{eq:app-P1}, since  the right-hand sides of \eqref{eq:app-P1} or \eqref{eq:app-logmetric-limit} are equal. 

We next show that, in the regime $p<\delta<\tfrac{\M-1}{\M}$,
\begin{equation} \label{eq:app-P2}
	\plimsup \frac1n \log\frac{1}{P_{\widetilde Y^n}(\Yn)}
	\leq \log \M.
\end{equation}
This will conclude the proof.
The proof is based on showing that the random variable $-\tfrac1n \log {P_{\widetilde{Y}^n}(\Yn)}$ has asymptotically vanishing mass beyond $\log \M$ as $n\to\infty$.
Therefore, for every $\epsilon>0$ we have
\begin{align}
	\Pr \biggl[ \frac1n \log\frac{1}{P_{\widetilde Y^n}(\Yn)} >	\log \M+\epsilon \biggr] & =	\Pr \biggl[ \frac{1}{P_{\widetilde Y^n}(\Yn)} > \M^n e^{n\epsilon} \biggr] \\
	&\quad\leq \frac{1}{\M^n e^{n\epsilon}} \sum_{\yn} \frac{P_{Y^n}(\yn)}	{P_{\widetilde Y^n}(\yn)} \label{eq:app_ineq1_markov}\\
	&\quad\leq e^{-n\epsilon}.
	\label{eq:app_ineq2_bound}
\end{align}
Inequality \eqref{eq:app_ineq1_markov} follows from applying Markov's inequality, whereas \eqref{eq:app_ineq2_bound} is obtained from upper-bounding the sum over $\yn$ by $\M^n$, as shown below.

The proof of the upper bound
\begin{equation}\label{eq:app-ratio-bound}
	\sum_{\yn} \frac{P_{Y^n}(\yn)} {P_{\widetilde Y^n}(\yn)} \leq \M^n
\end{equation}
relies on interpreting the decoding metric $\qn$ as a degraded version of channel $\Wn$. Specifically, within the familiy of $\M$-ary symmetric channels, we factor the metric as $\qn = \Wn V^n V^n$ for some channel $V^n$.
For $\W=\operatorname{KSC}(p)$ and $\q=\operatorname{KSC}(\delta)$, with \(p<\delta<(\M-1)/\M\), we have
\begin{equation}
	V = \operatorname{KSC}\left( \frac{\M-1}{\M} \left( 1-\sqrt{1-\frac{\M\delta}{\M-1}}	\right)\right).
\end{equation}

Consequently, we introduce the Markov chain
\begin{equation}
	\Xn \xrightarrow{\Wn} \Yn \xrightarrow{V^n}	Z^n \xrightarrow{V^n} \widetilde Y^n
\end{equation}
whose end-to-end transition law is $q^n$. Then, the following expressions are equivalent:
\begin{align}
	P_{\widetilde{Y}^n}(\yn) &= \sum_{\xn} \PXn(\xn) \qn(\xn, \yn) \\
	&= \sum_{z^n} P_{Z^n}(z^n) V^n(\yn|z^n) \\
	&= \sum_{z^n} P_{Z^n}(z^n) V^n(z^n|y^n). \label{eq:app-Pyt3}
\end{align}
The last equality follows from the symmetry of channel $V^n$ since  $V^n(\yn|z^n)=V^n(z^n|\yn)$. Indeed, observe that $\M$-ary symmetric channels are the same when used from input to output and from output to input.
Using \eqref{eq:app-Pyt3} we obtain
\begin{align} \label{eq:app-ratio-sum-1}
	\sum_{\yn} \frac{P_{\Yn}(\yn)} {P_{\widetilde Y^n}(\yn)}
	&= \sum_{\yn} \frac{P_{\Yn}(\yn)}{\sum_{z^n} P_{Z^n}(z^n) V^n(z^n|\yn)} \\ \label{eq:app-ratio-sum-2}
	&\leq \sum_{\yn\!,z^n} \frac{P_{\Yn}(\yn) V^n(z^n|\yn)}{P_{Z^n}(z^n)}\\ \label{eq:app-ratio-sum-3}
	&= \sum_{z^n} \frac{P_{Z^n}(z^n)}{P_{Z^n}(z^n)}\\
	&= \M^n
\end{align}
where \eqref{eq:app-ratio-sum-1} follows by using \eqref{eq:app-Pyt3}; \eqref{eq:app-ratio-sum-2} follows from applying Jensen's inequality interpreting $V^n(\cdot|\yn)$ as a weighting distribution; \eqref{eq:app-ratio-sum-3} follows from  $\sum_{\yn} P_{\Yn}(\yn) V^n(z^n|\yn) = P_{Z^n}(z^n)$, according to the above Markov chain.

\subsection{$\frac{\M-1}{\M}\leq\delta<1$:}
Although channel and metric share the same symmetric structure, the channel favours correct transitions whereas the metric favours crossover transitions.
In this case, $C_{\text{GMI}}^{\text{sd},(1)}=0$, and, according to Proposition \ref{prop:nonzero}, it follows that $C_{\q}^{\text{sd}}(\W)=0$. Indeed, in this regime we have that $C_{\q}^{\text{mm}}(\W)=0$ and thus  
$C_{\q}^{\text{mm}}(\W)\geq C_{\q}^{\text{sd}}(\W)=0$.

\section{Proof for Binary Alphabets} \label{app:BIBO}
This appendix proves that, for binary alphabets $\Xc=\Yc=\{0,1\}$, the mismatch capacity under stochastic decoding yields the input-optimized LM rate, i.e., $C_{\q}^{\text{sd}}(\W) = C_{\text{LM}}^{\text{sd},(1)}$.
Rather than using equation \eqref{eq:LMk_intro} for $k=1$, we adopt a more convenient expression in terms of the channel induced by the decoding metric
\begin{equation} \label{eq:app_barqs}
	\bar{\q}_s(y|x) \triangleq \frac{\q(x, y)^s e^{a_\star(x)} P_Y(y)}{\EX[\q(\bar{X}, y)^s e^{a_\star(\bar{X})}]}%
\end{equation}
with $a_\star(x)$ chosen so that $\sum_y \bar{q}_s(y|x) = 1$ for $x\in\{0,1\}$, as
\begin{equation} \label{eq:app_CLM2}
	C_{\q}^{\text{sd}} (\W) = \sup_{P_{X}} \sup_{s\in(0, 1]} \EX \biggl[ \log \frac{\bar{q}_s(Y|X)}{P_Y(Y)} \biggr].%
\end{equation}
We restrict attention to channel-metric pairs satisfying the condition in \eqref{eq:lemmaUIcond}, i.e., 
\begin{align}
\min_{x,y:W(y|x)>0} q(x,y)>0.
\end{align}

We first review the corresponding result for maximum metric decoding. We then establish the properties needed to bound the multi-letter expression \eqref{eq:CM_DMCadd}, and prove the result for stochastic decoding in Appendix \ref{app:review_sd}.

\subsection{Review of maximum metric decoding} \label{app:review_mm}
Interpreting $W$ and $q$ as matrices, the mismatch capacity for binary alphabets \cite[Ex. 6]{370120}
\begin{equation}
	C_{\q}^{\text{mm}}(\W) =
	\begin{cases}
		C(\W) & \text{if } \det(\W) \det(\q) >0 \\
		0     & \text{if } \det(\W) \det(\q) \leq0.
	\end{cases}
\end{equation}
This is achieved by random constant-composition codes through the LM rate. When the metric is aligned with the channel $\det(\W) \det(\q) >0$, $C_{\q}^{\text{mm}}(\W)=C(W)$. Instead, when the metric is misaligned with the channel $\det(\W) \det(\q) \leq 0$, $C_{\q}^{\text{mm}}(\W)=0$. 
In particular, when $\det(\W)\det(\q)>0$, parameter $s$ in the LM rate can be chosen in terms of the ratio of the channel and metric matrices, as%
\begin{equation} \label{eq:app:sopt}
	s_\star \triangleq \left( \log \frac{\W_{00}\W_{11}}{\W_{01}\W_{10}} \right) \left( \log \frac{\q_{00}\q_{11}}{\q_{01}\q_{10}} \right)^{-1} > 0.%
\end{equation}
For this choice, we have that the channel induced by the metric \eqref{eq:app_barqs} is such that $\bar{\q}_{s_\star}(y|x) = \W(y|x)$, thereby turning the maximum metric decoding into maximum likelihood.

\subsection{Properties} \label{app:properties}
\begin{enumerate}[label=\arabic*)]
	\item \label{item_prop_s}
	The optimization over $s>0$ in \eqref{eq:app_CLM2} is concave in $s$. Therefore, if $s_\star>1$, the supremum in \eqref{eq:app_CLM2}, restricted to $s\in(0,1]$, is attained for $s=1$.
	
	\item \label{item_prop_Py}
	For every $s>0$, $\W$ and $\bar{\q}_s$ induce the same output marginal under input $P_X$. From \eqref{eq:app_barqs}, we have that $P_Y(y) = (P_X\W)(y) = (P_X\bar{\q}_s)(y)$.
	
	\item \label{item_prop_OR}
	Since the ratio in \eqref{eq:app:sopt} is invariant to positive row and column scaling of the matrices $W,q$, we have
	\begin{equation}
		\left( \frac{\q_{00}\q_{11}}{\q_{01}\q_{10}} \right)^s = \frac{\bar{q}_{s,00} \bar{q}_{s,11}}{\bar{q}_{s,01}\bar{q}_{s,10}}.%
	\end{equation}
	Using $s_\star$ in \eqref{eq:app:sopt}, we additionally have%
	\begin{equation}
		\frac{\W_{00}\W_{11}}{\W_{01}\W_{10}} = \left( \frac{\q_{00}\q_{11}}{\q_{01}\q_{10}} \right)^{s_\star}=  \left( \frac{\bar{q}_{s,00} \bar{q}_{s,11}}{\bar{q}_{s,01}\bar{q}_{s,10}} \right)^{s_\star/s}\!.%
		\label{eq:app_prop3_proof_qW}
	\end{equation}
	
	\item \label{item_prop_det}
	Suppose $s_\star>1$. Then%
	\begin{equation}
		0 < \frac{\det(\bar{q}_s)}{\det (\W)} < 1.%
	\end{equation}
	
	By property \ref{item_prop_Py}, from $(P_X\W)(y)=(P_X\bar{\q}_s)(y)$ we have
	\begin{equation} \label{eq:app-Py}
		P_X(0) (\W_{00} - \bar{\q}_{s, 00}) + P_X(1)(\W_{10} - \bar{\q}_{s, 10}) = 0.
	\end{equation}
	\begin{itemize}
		\item Assume $\W_{00} > \bar{q}_{s,00}$. Then, $\W_{10}< \bar{q}_{s,10}$ follows from \eqref{eq:app-Py}. These imply%
		\begin{equation} \label{eq:app-comb0}
			\frac{\W_{00}\W_{11}}{\W_{01}\W_{10}} >  \frac{\bar{q}_{s,00} \bar{q}_{s,11}}{\bar{q}_{s,01}\bar{q}_{s,10}}.
		\end{equation}
		Combining this with property \ref{item_prop_OR} we have%
		\begin{equation} \label{eq:app-comb}
			\left( \frac{\bar{q}_{s,00} \bar{q}_{s,11}}{\bar{q}_{s,01}\bar{q}_{s,10}} \right)^{s_\star/s} >  \frac{\bar{q}_{s,00} \bar{q}_{s,11}}{\bar{q}_{s,01}\bar{q}_{s,10}}
		\end{equation}
		for any $s>0$. For our case of interest, since $s_\star/s>1$ for $s\in(0,1]$, \eqref{eq:app-comb} is only satisfied if%
		\begin{equation}
			\frac{\bar{q}_{s,00} \bar{q}_{s,11}}{\bar{q}_{s,01}\bar{q}_{s,10}} > 1.
			\label{eq:prop4_app_q}
		\end{equation}
		Since $\det(\bar{\q}_s)= \bar{q}_{s,00} \bar{q}_{s,11} - \bar{q}_{s,01} \bar{q}_{s,10}$,  condition \eqref{eq:prop4_app_q} becomes $\det (\bar{\q}_s) >0$.	
		Together with \eqref{eq:app-comb} and \eqref{eq:app_prop3_proof_qW}, we have $\det(\W) > \det (\bar{\q}_s) >0$.
		
		\item Assume now that $\W_{00} < \bar{q}_{s,00}$. By the same reasoning as the point above, we have $\W_{10} > \bar{q}_{s,10}$ and
		\begin{equation}
			\left( \frac{\bar{q}_{s,00} \bar{q}_{s,11}}{\bar{q}_{s,01}\bar{q}_{s,10}} \right)^{s_\star/s} <  \frac{\bar{q}_{s,00} \bar{q}_{s,11}}{\bar{q}_{s,01}\bar{q}_{s,10}}
		\end{equation}
		for any $s>0$. For $s\in(0,1]$, the former is only satisfied if the right-hand side is less than $1$. This implies $\det(\W) < \det (\bar{\q}_s) < 0$.
	\end{itemize}
	
	\item \label{item_prop_degraded}
	Suppose $s_\star>1$ and $P_X$ with full support. Then, similarly to the $\M$-ary symmetric channel case, the channel conditional distribution $\bar{\q}_1$ admits the stochastic degradations 
	\begin{equation} \label{eq:app-degraded}
		\begin{tikzcd}[column sep=2.0em, row sep=0em]
			X \arrow[rrr, "\bar{\q}_1"] & & & \widetilde{Y} \\
			X \arrow[r, "\W"] & Y \arrow[rr, "T"] & & \widetilde{Y} \\
			X \arrow[r, "\W"] & Y \arrow[r, "V"]  & Z \arrow[r, "V"] & \widetilde{Y}
		\end{tikzcd}
	\end{equation}
	where $\bar{\q}_1$, $\W$, $T$, and $V$ are binary-input binary-output channels.\\
	
	To prove this:
	\begin{itemize}
		\item Define the transition matrix $T\triangleq\W^{-1}\bar{\q}_1$; here $(\cdot )^{-1}$ denotes matrix inversion. We want to show that $T$ defines a valid conditional channel distribution. Since both $\W$ and $\bar{\q}_1$ have unit row sums,
		\begin{equation}
			\sum_{\tilde y\in\{0,1\}} T(\tilde{y}| y)=1\quad \text{for}\quad y\in\{0,1\}.
		\end{equation}
		We now prove the non-negativity of $T$. By property \ref{item_prop_Py}, we have
		\begin{equation}
			(P_Y T)(\tilde{y}) = (P_X\bar{\q}_1)(\tilde{y}) = P_Y(\tilde{y}).
		\end{equation}
		That is, $P_Y$ is a stationary distribution of $T$. Hence, using the spectral decomposition of $T$ (see \cite{leon2004optimal}),%
		\begin{equation} \label{eq:app-Tdef}
			T(\tilde{y}|y) = \lambda \IND{\tilde{y}=y} + (1-\lambda)P_Y(\tilde{y})
		\end{equation}
		where $\lambda = \det(T)$.
		If $0<\det(T)<1$, then $T(\tilde{y}|y)>0$ since $1-\lambda>0$ and $0<P_Y(\tilde{y})<1$.
		By property \ref{item_prop_det} when $s_\star>1$, we have $0< \det(T) = \frac{\det(\bar{\q}_1)}{\det(\W)}<1$. Consequently, since $T$ has non-negative entries and its rows sum to one, it defines a valid channel.
		
		\item Define%
		\begin{equation} \label{eq:app-Vdef}
			V(z|y) \triangleq \sqrt{\lambda} \IND{z=y} + (1-\sqrt{\lambda}) P_Y(z).
		\end{equation}
		The claim follows by checking that%
		\begin{equation}
			\sum_{z\in\{0,1\}} V(\tilde{y}|z) V(z|y) = T(\tilde{y}|y).
		\end{equation}
	\end{itemize}
\end{enumerate}

The following properties concern the degradation in \eqref{eq:app-degraded}:
\begin{enumerate} [resume, label=\arabic*)]
	\item \label{item_prop_Pystat}
	$P_Y$ is a stationary distribution of channels $T$ and $V$. The proof follows from the spectral decompositions of $T$ and $V$ in \eqref{eq:app-Tdef}--\eqref{eq:app-Vdef}.
	
	\item \label{item_prop_reverse}
	The following holds for $\tilde{y},z \in \{0,1\}$ \cite[p. 19]{fan2021hoeffding}:
	\begin{equation}
		P_Y(z) V(\tilde{y}|z) = P_Y(\tilde{y}) V(z|\tilde{y}).
	\end{equation}
	This condition is the extension of the symmetry property used in \eqref{eq:app-Pyt3}.
\end{enumerate}

The above list or properties serve as a prelude to the following proposition and corollary, which are the main ingredients of the proof of the main result shown in the next section.
\begin{proposition} \label{item_prop_bound}
	For a stochastically degraded channel $\bar{\q}_1$ satisfying \eqref{eq:app-degraded}, it holds for any input distribution $P_{\Xn}$ that%
	\begin{equation} \label{eq:app-lemmaJ}
		\sum_{\yn} P_Y^{\otimes n}(\yn) \frac{(P_{\Xn} \Wn)(\yn)}{(P_{\Xn} \bar{\q}_1^n)(\yn)} \leq 1.
	\end{equation}
	$P_Y^{\otimes n}(\yn) = \prod_{i=1}^n P_Y(y_i)$ is the $n$-letter product distribution of single-letter output marginals. 
	$\Wn$ and $\bar{q}_1^n$ are defined analogously for the single-letter channels $\W$ and $\bar{q}_1$.
\end{proposition}
\begin{proof}
	Using $\PYn(\yn) = (\PXn\Wn)(\yn)$ and noting from \eqref{eq:app-degraded} that $(P_{\Xn} \bar{\q}_1)(\yn) = (P_{Z^n} V^n)(\yn)$, we write the left-hand side of \eqref{eq:app-lemmaJ} as%
	\begin{align} \label{eq:app-Jdef}
		 \sum_{\yn} P_Y^{\otimes n}(\yn) \frac{(P_{\Xn} \Wn)(\yn)}{(P_{\Xn} \bar{\q}_1^n)(\yn)} 
		&=\sum_{\yn} P_Y^{\otimes n}(\yn) \frac{P_{\Yn}(\yn)}{{\sum_{z^n}} P_{Z^n}(z^n) V^n(\yn|z^n)}\\
		&\label{eq:app-J_1}= \sum_{\yn} \frac{P_{\Yn}(\yn)}{{\sum_{z^n}}  \frac{P_{Z^n}(z^n)}{P_Y^{\otimes n}(\yn)} \cdot V^n(\yn|z^n)}\\ \label{eq:app-J_2}
		&= \sum_{\yn} \frac{P_{\Yn}(\yn)}{{\sum_{z^n}}  \frac{P_{Z^n}(z^n)}{P_Y^{\otimes n}(z^n)} \cdot V^n(z^n|\yn)}\\ \label{eq:app-J_3}
		&\leq \sum_{\yn, z^n} \frac{P_{Y^n}(\yn) V^n(z^n|\yn)}{P_{Z^n}(z^n)/P_Y^{\otimes n}(z^n)}
	\end{align}
	where the equality in \eqref{eq:app-J_1} follows from rewriting \eqref{eq:app-Jdef}, equality \eqref{eq:app-J_2} uses property \ref{item_prop_reverse} on each of the $n$ channel uses, and eq. \eqref{eq:app-J_3} follows from applying Jensen's inequality in the denominator of \eqref{eq:app-J_2}, using $V^n(\cdot|\yn)$ as the weighting distribution. The numerator of \eqref{eq:app-J_3} is, in fact, ${\sum_{\yn}} \PYn(\yn)V^n(z^n|\yn) = P_{Z^n}(z^n)$, which gives
	\begin{align} \label{eq:app-J_4}
		{\sum_{\yn, z^n}} \frac{P_{Y^n}(\yn) V^n(z^n|\yn)}{P_{Z^n}(z^n)/P_Y^{\otimes n}(z^n)} &= \sum_{z^n} \frac{P_{Z^n}(z^n)}{P_{Z^n}(z^n)/P_Y^{\otimes n}(z^n)} \\
		&= \sum_{z^n} P_Y^{\otimes n}(z^n) \\
		&=1.
	\end{align}
	This concludes the proof.
\end{proof}

\begin{corollaryProp} \label{item_propCor_bound}
	For any input distribution $\PXn$ and $(\Xn, \Yn, \Xbn) \thicksim \PXn\Wn\PXn$, we have%
	\begin{equation}
		\EX\bigg[ \log \frac{e^{- a^n(\Xn)} P_Y^{\otimes n}(\Yn)}{\EX[e^{-a^n(\Xbn)} \bar{\q}_1^n(\Yn|\Xbn)]} \bigg] \leq 0.
	\end{equation}
\end{corollaryProp}
\begin{proof} Applying Jensen's inequality on the left-hand side and expanding expectations as summations, we have%
	\begin{align}
		\EX\bigg[ \log \frac{e^{- a^n(\Xn)} P_Y^{\otimes n}(\Yn)}{\EX[e^{-a^n(\Xbn)} \bar{\q}_1^n(\Yn|\Xbn)]} \bigg] 
		  &\leq \log \EX\bigg[ \frac{e^{- a^n(\Xn)} P_Y^{\otimes n}(\Yn)}{\EX[e^{-a^n(\Xbn)} \bar{\q}_1^n(\Yn|\Xbn)]} \bigg] \\
		  &= {\log}  {\sum_{\xn, \yn}} \frac{\PXn(\xn)\Wn(\yn|\xn) e^{- a^n(\xn)} P_Y^{\otimes n}(\yn)}{{\sum_{\bar{x}^n}} \PXn(\bar{x}^n) e^{-a^n(\bar{x}^n)} \bar{\q}_1^n(\yn|\bar{x}^n)} .
	\end{align}

	We now write the last equation in terms of the tilted distribution
	\begin{equation}
		\tilde{P}_{\Xn}(\xn) \triangleq \frac{\PXn(\xn) e^{-a^n(\xn)}}{\sum_{\bar{x}^n} \PXn(\bar{x}^n) e^{- a^n(\bar{x}^n)}}
	\end{equation}
	obtaining
	\begin{align}
		\EX\bigg[ \log \frac{e^{- a^n(\Xn)} P_Y^{\otimes n}(\Yn)}{\EX[e^{-a^n(\Xbn)} \bar{\q}_1^n(\Yn|\Xbn)]} \bigg]  &\leq {\log}  {\sum_{\xn, \yn}} \frac{\tilde{P}_{\Xn}(\xn)\Wn(\yn|\xn)  P_Y^{\otimes n}(\yn)}{{\sum_{\bar{x}^n}} \tilde{P}_{\Xn}(\bar{x}^n) \bar{\q}_1^n(\yn|\bar{x}^n)} \\
		& =	\log\sum_{\yn}P_Y^{\otimes n}(\yn)
		\frac{(\tilde{P}_{\Xn}\Wn)(\yn)} 
		{(\tilde{P}_{\Xn}\bar{\q}_1^n)(\yn)} \label{eq:app_jensens1_bound}\\
		&\leq 0, \label{eq:app_final_bound}
	\end{align}
	where the last step \eqref{eq:app_final_bound} uses Proposition \ref{item_prop_bound}. This concludes the proof.
\end{proof}

\subsection{Stochastic decoding} \label{app:review_sd}
The mismatch capacity under stochastic decoding equals the input-optimized LM rate, which can be written as
\begin{equation}
	C_{\q}^{\text{sd}}(\W) =
	\begin{cases}
		C(\W) & \text{if } 0< s_\star \leq 1 \\
		\sup_{P_X} \EX \Bigl[ \log \frac{\bar{q}_1(Y|X)}{P_Y(Y)} \Bigr] & \text{if } s_\star > 1 \\
		0     & \text{if } s_\star \leq 0
	\end{cases}
\end{equation}
where $s_\star$ is defined in \eqref{eq:app:sopt}.

\paragraph{Regime $0<s_\star\leq1$}
It follows directly from the maximum metric argument. The metric can be aligned with the channel through appropriate tilting and normalization. Consequently, $\bar{\q}_{s_{\star}} = \W$ and
\begin{equation}
	C_{\q}^{\text{sd}}(\W) = C_{\q}^{\text{mm}}(\W) = \sup_{P_X} I(X;Y).
\end{equation}

\paragraph{Regime $s_\star > 1$} Since the range of $s$ for the stochastic decoder is limited to $s\in(0,1]$, the decoding metric cannot be fully aligned with the channel, as is the case for maximum metric decoding. Rather, by property \ref{item_prop_s} the input-optimized LM rate is achieved for $s=1$, yielding%
\begin{equation}
	C_{\text{LM}}^{\text{sd},(1)} (\W) = \sup_{P_{X}} \EX \biggl[ \log \frac{\bar \q_1(X, Y)}{P_Y(Y)} \biggr].
\end{equation}

We prove below that it corresponds to the mismatch capacity. 
Our analysis is based on a single-letter upper bound to the multi-letter expression in Theorem \ref{Th:CM_DMCadd}%
\begin{equation} \label{eq:app-Cqsd}
	C_{\q}^{\text{sd}}(\W) = \lim_{k\to\infty} \sup_{P_{X^k}} \frac{1}{k} \EX \biggl[ \log \frac{\q^k(X^k, Y^k)}{\EX[ \q^k(\bar{X}^k, Y^k)|Y^k]} \biggr].
\end{equation}
Specifically, we first decompose the right-hand side without the factor $1/k$ as%
\begin{equation}
	\EX \biggl[ \log \frac{\q^k(X^k, Y^k)}{\EX[ \q^k(\bar{X}^k, Y^k)|Y^k]} \biggr] = \EX \biggl[ \log \frac{\bar \q_1^k(Y^k|X^k)}{P_{Y}^{\otimes k}(Y^k)} \biggr] 
	+ \EX \biggl[ \log \frac{e^{-\sum_i a_\star(X_i)} P_{Y}^{\otimes k}(Y^k)}{\EX[ \bar{\q}_1^k(Y^k|\bar{X}^k) e^{-\sum_i a_\star(\bar{X}_i)}|Y^k]} \biggr].
\end{equation}
Using Corollary~\ref{item_propCor_bound} to the second term, we have%
\begin{align}
	\!\EX \biggl[ \log \frac{\q^k(X^k, Y^k)}{\EX[ \q^k(\bar{X}^k, Y^k)|Y^k]} \biggr] &\leq \EX \biggl[ \log \frac{\bar \q_1^k(Y^k| X^k)}{P_{Y}^{\otimes k}(Y^k)} \biggr] \\
	&= \sum_{i=1}^k \EX \biggl[ \log \frac{\bar \q_1(Y_i| X_i)}{P_{Y}(Y_i)} \biggr] \\
	&= \sup_{P_X} k \EX \biggl[ \log \frac{\bar \q_1(Y|X)}{P_{Y}(Y)} \biggr]\!. 
\end{align}
Substituting this result into \eqref{eq:app-Cqsd} eliminates the limit and the optimization over the multi-letter $P_{X^k}$, yielding%
\begin{equation}
	C_{\q}^{\text{sd}} (\W) = C_{\text{LM}}^{\text{sd},(1)} (\W) = \sup_{P_{X}} \EX \biggl[ \log \frac{\bar \q_1(Y|X)}{P_Y(Y)} \biggr].
\end{equation}

\paragraph{Regime $s_\star \leq 0$}
By Appendix~\ref{app:review_mm}, $C_{\q}^{\text{mm}}(\W)=0$. Hence $C_{\q}^{\text{sd}}(\W) \leq C_{\q}^{\text{mm}}(\W) = 0$.

\bibliographystyle{IEEEtran}
\bibliography{finalreport}
\end{document}